\documentclass[11pt,a4paper]{article}
\usepackage{amssymb}
\usepackage{amsmath}
\usepackage{amsfonts}
\usepackage{bbm}
\usepackage{amsthm}
\usepackage{mathrsfs}
\usepackage{hyperref}
\usepackage{color}
\usepackage[margin=2.41cm]{geometry}
\usepackage[all,cmtip]{xy}
\usepackage[utf8]{inputenc}
\usepackage{graphicx}
\usepackage{varwidth}
\usepackage{comment}
\usepackage{enumitem}

\usepackage{mathtools}

\usepackage{upgreek}
\usepackage{rotating}

\usepackage{tikz}
\usetikzlibrary{shapes.geometric}

\definecolor{darkred}{rgb}{0.8,0.1,0.1}
\hypersetup{
	colorlinks=true,         
	linkcolor=darkred,
	citecolor=blue,
}

\theoremstyle{plain}
\newtheorem{theo}{Theorem}[section]
\newtheorem{lem}[theo]{Lemma}
\newtheorem{propo}[theo]{Proposition}
\newtheorem{cor}[theo]{Corollary}

\theoremstyle{definition}
\newtheorem{defi}[theo]{Definition}

\newtheorem{exx}[theo]{Example}
\newenvironment{ex}
{\pushQED{\qed}\exx}
{\popQED\endexx}

\newtheorem{remm}[theo]{Remark}
\newenvironment{rem}
{\pushQED{\qed}\remm}
{\popQED\endremm}

\newtheorem{constrr}[theo]{Construction}
\newenvironment{constr}
{\pushQED{\qed}\constrr}
{\popQED\endconstrr}

\numberwithin{equation}{section}

\def\nn{\nonumber}

\def\bbR{\mathbb{R}}
\def\bbC{\mathbb{C}}

\def\ii{{\,{\rm i}\,}}

\def\id{\mathrm{id}}
\def\supp{\mathrm{supp}}

\def\vol{\mathrm{vol}}

\def\sc{\mathrm{sc}}
\def\cc{\mathrm{c}}

\def\1{I}
\def\oone{\mathbbm{1}}
\def\op{\mathrm{op}}

\def\Loc{\mathbf{Loc}}
\def\Cau{\mathbf{Cau}}

\def\Alg{\mathbf{Alg}}

\def\Vec{\mathbf{Vec}}

\def\CC{\mathbf{C}}
\def\DD{\mathbf{D}}

\def\TT{\mathbf{T}}
\def\Cat{\mathbf{Cat}}
\def\PsCat{\mathbf{PsCat}}

\def\Grpd{\mathbf{Grpd}}

\def\AQFT{\mathbf{AQFT}}
\def\FFT{\mathbf{FFT}}

\def\Fun{\mathbf{Fun}}
\def\PsFun{\mathbf{PsFun}}
\def\LB{\mathcal{LB}\mathsf{ord}}

\def\res{\operatorname{res}}

\def\AAA{\mathfrak{A}}
\def\LLL{\mathfrak{L}}
\def\BBB{\mathfrak{B}}

\def\FFF{\mathfrak{F}}
\def\Sol{\mathfrak{Sol}}
\def\Data{\mathfrak{Data}}
\def\CCR{\mathfrak{CCR}}

\def\C{\mathcal{C}}
\def\D{\mathcal{D}}

\def\KG{\mathrm{KG}}

\def\colim{\mathrm{colim}}

\def\sla{{\scriptscriptstyle \slash}}

\DeclareMathOperator*{\cmp}{\text{\raisebox{0.25ex}{\scalebox{0.7}{$\odot\,$}}}}

\def\sk{\vspace{2mm}}

\def\hto{\nrightarrow}

\makeatletter
\let\@fnsymbol\@alph
\makeatother

\title{%
Lorentzian bordisms in algebraic quantum field theory
}

\author{%
Severin Bunk$^{1,a}$, James MacManus$^{2,b}$\ and\ Alexander Schenkel$^{2,c}$\vspace{4mm}\\
{\small ${}^1$ Department of Physics, Astronomy and Mathematics, University of Hertfordshire,}\\ 
{\small College Lane, Hatfield, AL10 9AB, United Kingdom.}\vspace{2mm}\\
{\small ${}^2$ School of Mathematical Sciences, University of Nottingham,}\\
{\small University Park, Nottingham NG7 2RD, United Kingdom.}\vspace{4mm}\\
{\small \begin{tabular}{ll}
Email: & ${}^a$~\href{mailto:s.bunk@herts.ac.uk}{\texttt{s.bunk@herts.ac.uk}}\\
& ${}^b$~\href{mailto:james.macmanus@nottingham.ac.uk}{\texttt{james.macmanus@nottingham.ac.uk}}\\
& ${}^c$~\href{mailto:alexander.schenkel@nottingham.ac.uk}{\texttt{alexander.schenkel@nottingham.ac.uk}}
\vspace{2mm}
\end{tabular}
}
}

\date{May 2024}

\begin{document}

\maketitle

\begin{abstract}
\noindent It is shown that every algebraic quantum field theory has an underlying functorial field theory which is defined on a suitable globally hyperbolic Lorentzian bordism pseudo-category. This means that globally hyperbolic Lorentzian bordisms between Cauchy surfaces arise naturally in the context of algebraic quantum field theory. The underlying functorial field theory encodes the time evolution of the original theory, but not its spatially local structure. As an illustrative application of these results, the algebraic and functorial descriptions of a free scalar quantum field are compared in detail.
\end{abstract}
\vspace{-1mm}

\vspace{-2mm}

\vspace{-2mm}

\tableofcontents

\newpage 


\section{\label{sec:intro}Introduction and summary}
The mathematical axiomatization of quantum field theory (QFT) is a 
long-standing and important problem in mathematical physics, dating back to
the early 1950s when Wightman proposed a first set of axioms for relativistic QFT.
Over the subsequent decades, the field of mathematical QFT went through
substantial developments, which to date have manifested themselves into
a variety of approaches to this subject, including most notably: 1.)~Functorial QFT (FFT), going back to ideas
of Witten, Atiyah and Segal \cite{Witten,Atiyah,Segal}, 2.)~algebraic QFT (AQFT),
initiated by Haag and Kastler \cite{HaagKastler}, and 3.)~factorization algebras,
developed by Costello and Gwilliam \cite{CG1,CG2}. 
\sk

The main idea behind FFT \cite{Witten,Atiyah,Segal} is to describe a QFT in terms of a functor
from a bordism category to a suitable target category, for instance that of vector spaces. 
The typical interpretation of such a functor is as the assignment of
a state space to each $(m{-}1)$-dimensional manifold and of a linear map
to each $m$-dimensional bordism.
These linear maps encode a concept of `time evolution' 
along the bordisms. Since its inception, the field of FFT has gone through 
immense internal developments, which proceeded in parallel with the development of novel
higher categorical structures in pure mathematics. Most notable are the proposal
of \textit{extended} FFTs \cite{BaezDolan} which assign higher categorical
data to manifolds of any codimension, see also \cite{Lurie,CalaqueScheimbauer}
for the mathematical foundations, and the concept of \textit{geometric} FFTs 
\cite{StolzTeichner} in which the bordisms are endowed with additional geometric structures, 
such as a metric. Concrete examples of the latter have been constructed in 
\cite{BW:Transgression,BW:OCFFTs,LudewigStoffel,Kandel,Kandel2} and their extended versions
have been developed and studied in \cite{GradyPavlov,GradyPavlov2}.
\sk

The main idea behind AQFT \cite{HaagKastler} and factorization algebras \cite{CG1,CG2}
is crucially different from that of FFT. In contrast to state spaces, 
these approaches focus on the observables of a QFT, which are assigned 
locally to suitable open subsets of spacetime. In physics terminology,
one may say that these approaches focus on the Heisenberg picture of quantum theory,
while FFT typically emphasizes the Schr\"odinger picture. AQFT
has also gone through substantial internal developments over the past decades,
leading in particular to a locally covariant version \cite{BFV,FV}, defined
on all globally hyperbolic Lorentzian manifolds in contrast to only the 
Minkowski spacetime as in the original work \cite{HaagKastler}, the incorporation
of powerful operadic techniques \cite{operadAQFT}, and the development of 
homotopical and higher categorical generalizations \cite{homotopyAQFT,2AQFT}
that are relevant for the description of gauge theories \cite{FR1,FR2}.
\sk

Having available multiple axiomatizations of QFT triggers a series of important questions,
in particular whether different axiomatizations are compatible with each other and, 
if that is the case, how they can be related.
The relationship between AQFT and factorization algebras, in the 
relativistic context where the geometric side consists of
globally hyperbolic Lorentzian spacetimes, is by now well understood
due to the works of Gwilliam and Rejzner \cite{GR1,GR2}
and of Benini, Musante, Perin and the third author \cite{FAvsAQFT,FreeFAvsAQFT}.
The main results of these papers are an equivalence theorem between 
AQFTs and factorization algebras on globally hyperbolic 
Lorentzian manifolds (subject to some mild additional conditions),
and a dictionary relating examples on both sides. In contrast to this, the relationship between
FFT and AQFT or factorization algebras is currently 
less investigated and understood. The only works in this direction which we are aware of
are the paper \cite{Schreiber} by Schreiber, which shows that taking
endomorphism algebras of an extended FFT defines an AQFT,
and the paper \cite{JohnsonFreyd} by Johnson-Freyd, which presents some general
ideas on how to pass between the Schr\"odinger and Heisenberg picture 
in QFT.
\sk

The main goal of the present paper is to show that every AQFT
has an underlying FFT which is defined on a globally hyperbolic
Lorentzian version of the Stolz-Teichner geometric bordism pseudo-categories 
\cite{StolzTeichner}. Our results show that globally hyperbolic Lorentzian 
bordisms appear naturally in AQFT and that they capture, 
in a precise sense, exactly those parts of an AQFT that are related
to time evolution. Let us explain our results in more precise terms. 
In Theorem~\ref{theo:comparison1}, we construct
a functor $\FFF_{(-)}^{} : \AQFT_m\to \FFT_m^{\mathrm{t.s.}}$ from the groupoid
of $m$-dimensional AQFTs (see Definition \ref{def:AQFT})
to the groupoid of $m$-dimensional globally hyperbolic Lorentzian FFTs 
(see Definition \ref{def:FFT}) that satisfy the time-slice axiom.
It is important to emphasize that our FFTs take values in algebras, 
as do AQFTs, which means that Theorem \ref{theo:comparison1} does 
{\it not} describe a transition between the Heisenberg and Schr\"odinger picture.
Our comparison functor has a geometric origin, which lies in the similarities
between the category $\Loc_m$ of globally hyperbolic Lorentzian 
spacetimes used in AQFT (see Definition \ref{def:Loc}) and the globally hyperbolic
Lorentzian bordism pseudo-category $\LB_m$, which we develop in Section \ref{sec:LBord}.
The functor $\FFF_{(-)}^{} : \AQFT_m\to \FFT_m^{\mathrm{t.s.}}$ from Theorem \ref{theo:comparison1}
is faithful, but in spacetime dimension $m\geq 2$, it fails to be full and 
essentially surjective. This means that 
passing from an AQFT $\AAA\in\AQFT_m$ to its underlying FFT
$\FFF_{\AAA}^{}\in \FFT_m^{\mathrm{t.s.}}$ is forgetful. 
We identify those parts of the AQFT that the functor $\FFF_{(-)}^{}$ forgets
with its \textit{spatially local} structure, given by morphisms in $\Loc_m$ which
are not Cauchy in the sense of Definition \ref{def:Locspecialmorphisms}. 
This interpretation is mathematically substantiated by Theorem \ref{theo:comparison2}, 
in which we show that forgetting the spatially local structure on the AQFT side, by restricting to
the subcategory $\Cau_m\subseteq \Loc_m$ consisting only of Cauchy morphisms, 
gives an equivalence between such \textit{spatially global} AQFTs and 
FFTs. Since the spatially local structure is a phenomenon that only occurs in spacetime dimension $m\geq 2$,
our results provide an equivalence $\AQFT_1\simeq \FFT_1^{\mathrm{t.s.}}$ in $m=1$ dimension, 
see Corollary \ref{cor:1d}. In future work, 
we intend to upgrade Theorem \ref{theo:comparison2} 
to a comparison theorem between QFTs with spatially local structures.
While the details of this generalization are currently not clear to us,
we would like to mention the following two potential approaches:
On the one hand, one could consider as in \cite[Remark 4.9]{FuchsSchweigertYang} and \cite[Section 9]{Yang}
richer bordism double categories than $\LB_m$ from Section \ref{sec:LBord} which also encode 
non-Cauchy morphisms in the vertical dimension. On the other hand, one could
try to generalize the Lorentzian bordism pseudo-category $\LB_m$ to a pseudo-operad
in order to encode spatially local structure through higher arity operations 
associated with disjoint subsets of a Cauchy surface.
\sk

The outline of the remainder of this paper is as follows.
In Section \ref{sec:prelim}, we recall some basic concepts from Lorentzian 
geometry and the theory of pseudo-categories that are required to state and prove our results.
In Section \ref{sec:LBord}, we define a globally hyperbolic Lorentzian variant of the
Stolz-Teichner geometric bordism pseudo-categories \cite{StolzTeichner}, which is needed
to introduce a globally hyperbolic Lorentzian variant of FFTs in Definition \ref{def:FFT}.
Our main results are contained in Section \ref{sec:comparison}. Theorem \ref{theo:comparison1}
shows that there exists a faithful functor $\FFF_{(-)}^{} : \AQFT_m\to \FFT_m^{\mathrm{t.s.}}$ 
that assigns to each AQFT an underlying globally hyperbolic Lorentzian FFT satisfying the time-slice axiom.
This functor is in general not an equivalence (unless in the very special case of $m=1$ spacetime
dimensions, see Corollary \ref{cor:1d}),
 but it forgets the spatially local structure of AQFTs. Theorem \ref{theo:comparison2}
then shows that, restricted to spatially global AQFTs, we obtain an equivalence.
In Section \ref{sec:example}, we test and illustrate our comparison results by studying
the simple example given by a free scalar quantum field. We describe this model both as an AQFT
and as an FFT, and then prove in Proposition \ref{prop:scalarcomparison} 
that the two different descriptions are compatible with each other through our comparison theorems.
This example also illustrates nicely in a concrete context that the underlying FFT of an AQFT encodes
precisely its time evolution. Appendix \ref{app:pseudocats} contains some technical results
about pseudo-categories that are needed in this paper.


\section{\label{sec:prelim}Preliminaries}

\subsection{Lorentzian geometry}
The aim of this section is to recall 
some basic definitions and properties of Lorentzian manifolds.
We refer the reader to \cite{BGP} for an introduction
and \cite{ONeill} for a comprehensive textbook.
\sk

A \textit{Lorentzian manifold} is a manifold $M$ 
(always assumed to be without boundary)
that is endowed with a metric $g$ of signature $(-++\cdots+)$. A non-zero tangent
vector $0\neq v\in T_pM$ at a point $p\in M$ is called
\textit{time-like} if $g(v,v)<0$,
\textit{light-like} if $g(v,v)=0$, and
\textit{space-like} if $g(v,v)>0$.
One says that the tangent vector $0\neq v\in T_pM$ 
is \textit{causal} if it is either time-like or light-like.
A smooth curve $\gamma : I\to M$, from an open or closed interval $I\subseteq \bbR$,
is called time-like/light-like/space-like/causal if all its tangent
vectors $\dot{\gamma}$ are time-like/light-like/space-like/causal.
\sk

A Lorentzian manifold $M$ is called \textit{time-orientable}
if there exists a vector field $\mathfrak{t}\in\Gamma^{\infty}(TM)$
that is everywhere time-like. A \textit{time-orientation}
is an equivalence class $[\mathfrak{t}]$ of time-like vector fields
for the following equivalence relation:
Two time-like vector fields $\mathfrak{t}$ and $\mathfrak{t}^\prime$
define the same time-orientation if $g(\mathfrak{t},\mathfrak{t}^\prime)<0$.
On a time-oriented Lorentzian manifold $M$, one can distinguish
between two types of time-like/causal curves: A time-like/causal curve
is called \textit{future directed} if $g(\mathfrak{t},\dot{\gamma})<0$
and it is \textit{past directed} if $g(\mathfrak{t},\dot{\gamma})>0$.
For any point $p\in M$, we introduce the \textit{chronological future/past} of $p$ as the subset
\begin{flalign} 
I^{\pm}_M(p)\,:=\, \left\{q\in M
\,\middle\vert\,
\begin{array}{c}
	\exists \text{ future/past directed time-like curve } \\
	\gamma:[0,1]\to M \text{ s.t.\ }\gamma(0)=p\text{ and } \gamma(1)=q
\end{array}
\right\}\,\subseteq\,M\quad,
\end{flalign}
and the \textit{causal future/past} of $p$ as the subset
\begin{flalign} 
J^{\pm}_M(p)\,:=\, \left\{q\in M
\,\middle\vert\,
\begin{array}{c}
	q=p \text{ or } \exists \text{ future/past directed causal curve } \\
	\gamma:[0,1]\to M \text{ s.t.\ }\gamma(0)=p\text{ and } \gamma(1)=q
\end{array}
\right\}\,\subseteq\,M\quad.
\end{flalign}
Given any subset $S\subseteq M$, we define 
\begin{flalign}
I^{\pm}_M(S)\,:=\, \bigcup_{p\in S} I^{\pm}_M(p)~~,\quad
J^{\pm}_M(S)\,:=\, \bigcup_{p\in S} J^{\pm}_M(p)\quad,
\end{flalign}
and
\begin{flalign}
I_M(S)\,:=\,I^{+}_M(S) \cup I^{-}_M(S)~~,\quad
J_M(S)\,:=\,J^{+}_M(S) \cup J^{-}_M(S)\quad.
\end{flalign}
The following definition introduces important types of
subsets in a time-oriented Lorentzian manifold.
\begin{defi}
Let $M$ be a time-oriented Lorentzian manifold.
\begin{itemize}
\item[(a)] A subset $S\subseteq M$ is called \textit{causally convex}
if $J^+_M(S)\cap J^-_M(S)\subseteq S$. In words, this means that 
every causal curve in $M$ that starts and ends in $S$ is entirely contained in $S$.

\item[(b)] Two subsets $S,S^\prime\subseteq M$ are called \textit{causally disjoint}
if $J_M(S)\cap S^\prime = \emptyset$, or equivalently $S\cap J_M(S^\prime) = \emptyset$.
In words, this means that there exists no causal curve in $M$ that connects $S$ and $S^\prime$.
\end{itemize}
\end{defi}

\begin{rem}\label{rem:causallyconvex}
Given two causally convex subsets $S_1,S_2\subseteq M$, 
their intersection $S_1\cap S_2\subseteq M$ is again causally convex. 
This fact will be used frequently in our work.
\end{rem}

A generic (time-oriented) Lorentzian manifold $M$ may have severe pathological features
such as closed time-like curves. These can be avoided by restricting to the following
well-behaved subclass of (time-oriented) Lorentzian manifolds.
\begin{defi}\label{def:globhyp}
A (time-oriented) Lorentzian manifold $M$ is called \textit{globally hyperbolic}
if it admits a \textit{Cauchy surface}, i.e.\ a subset $\Sigma\subset M$
that is met exactly once by each inextensible time-like curve in $M$.
\end{defi}

The following category of Lorentzian manifolds plays a fundamental role in AQFT,
see e.g.\ \cite{BFV,FV,operadAQFT,Grant-Stuart}.
\begin{defi}\label{def:Loc}
For any integer $m\geq 1$, the category $\Loc_m$ is defined by the following objects and morphisms:
\begin{itemize}
\item[\underline{$\mathsf{Obj}$:}] Oriented and time-oriented globally hyperbolic Lorentzian manifolds
$M$ (without boundary) of fixed dimension $m$.

\item[\underline{$\mathsf{Mor}$:}] Orientation and time-orientation preserving
isometric embeddings $f : M\to M^\prime$ whose image $f(M)\subseteq M^\prime$ is 
open and causally convex.
\end{itemize}
Composition is given by composition of maps, and the identities are the identity maps.
\end{defi}

\begin{rem}\label{rem:factorization}
Every $\Loc_m$-morphism $f:M\to M^\prime$ admits a factorization
\begin{flalign}
\begin{gathered}
\xymatrix{
\ar[dr]_-{f}^-{\cong}M\ar[rr]^-{f} && M^\prime\\
&f(M)\ar[ur]_-{\subseteq}&
}
\end{gathered}
\end{flalign}
in the category $\Loc_m$ into an isomorphism (denoted, by abuse of notation, by the same symbol)
followed by a subset inclusion. This fact will be used frequently in our work.
\end{rem}

The following distinguished types of morphisms in $\Loc_m$ play an important
role in AQFT and also in our present paper.
\begin{defi}\label{def:Locspecialmorphisms}
\begin{itemize}
\item[(a)] A $\Loc_m$-morphism $f:M\to M^\prime$ is called a \textit{Cauchy morphism}
if its image $f(M)\subseteq M^\prime$ contains a Cauchy surface of $M^\prime$.

\item[(b)] A pair of $\Loc_m$-morphisms $f_1 : M_1\to M \leftarrow M_2 : f_2$
to a common target is called \textit{causally disjoint} if the images
$f_1(M_1)\subseteq M$ and $f_2(M_2)\subseteq M$ are causally disjoint subsets of $M$.
\end{itemize}
\end{defi}

\subsection{Pseudo-categories}
The aim of this section is to recall some basic definitions 
from the theory of pseudo-categories, see e.g.\ \cite{PseudoCats}.
Pseudo-categories are a weak version of internal categories 
in a strict $2$-category. In the context of our paper,
the latter will be taken to be the strict $2$-category $\Grpd$
of groupoids, functors and natural transformations.
This means that our pseudo-categories are special cases
of (weak) double categories \cite{Shulman,Yau} in which
all vertical morphisms and $2$-cells are invertible.
In the context of quantum field theory, pseudo-categories 
are a convenient framework to define geometric bordism categories, see \cite{StolzTeichner}.
\begin{defi}\label{def:PsCat}
A \textit{pseudo-category} is a tuple
\begin{flalign}
\C\,=\,\big(\C_0,\C_1,s,t,\cmp,u,\mathsf{a},\mathsf{l},\mathsf{r}\big)
\end{flalign}
which consists of the following data:
\begin{itemize}
\item[(i)] A groupoid $\C_0$.

\item[(ii)] A span of groupoids and functors
\begin{flalign}
\begin{gathered}
\xymatrix@R=1em@C=1.5em{
&\ar[dl]_-{t}\C_1 \ar[dr]^-{s}&\\
\C_0&& \C_0
}
\end{gathered}\qquad.
\end{flalign}

\item[(iii)] A functor $\cmp : \C_1\times_{\C_0}\C_1\to \C_1$
from the (strict) fiber product 
\begin{flalign}
\begin{gathered}
\xymatrix@R=1em@C=1em{
&\ar@{-->}_-{p_1}[dl]\C_1\times_{\C_0}\C_1 \ar@{-->}[dr]^-{p_2}&\\
\ar[dr]_-{s}\C_1&&\C_1\ar[dl]^-{t}\\
&\C_0&
}
\end{gathered}\qquad
\end{flalign}
which is a map of spans
\begin{flalign}
\begin{gathered}
\xymatrix@R=1em@C=1em{
& \ar[dl]_-{t\,p_1}\ar[dd]_-{\cmp}\C_1\times_{\C_0}\C_1 \ar[dr]^-{s\,p_2}&\\
\C_0&&\C_0\\
&\ar[ul]^-{t}\C_1\ar[ur]_-{s}&
}
\end{gathered}\qquad,
\end{flalign}
where the triangles commute strictly.

\item[(iv)] A functor $u : \C_0\to \C_1$ which is a map of spans
\begin{flalign}
\begin{gathered}
\xymatrix@R=1em{
& \ar[dl]_-{\id}\ar[dd]_-{u}\C_0 \ar[dr]^-{\id}&\\
\C_0&&\C_0\\
&\ar[ul]^-{t}\C_1\ar[ur]_-{s}&
}
\end{gathered}\qquad,
\end{flalign}
where the triangles commute strictly.

\item[(v)] Natural isomorphisms $(\mathsf{a},\mathsf{l},\mathsf{r})$ that fill the following diagrams
of groupoids and functors
\begin{subequations}
\begin{flalign}
\begin{gathered}
\xymatrix@R=0.5em@C=1em{
\ar[dd]_-{\id\times_{\C_0}\cmp}\C_1\times_{\C_0}\C_1\times_{\C_0}\C_1 \ar[rr]^-{\cmp\times_{\C_0}\id}&& \C_1\times_{\C_0}\C_1\ar[dd]^-{\cmp}\\
&\rotatebox[origin=c]{45}{$\Leftarrow$}{\scriptstyle \mathsf{a}}& \\
\C_1\times_{\C_0}\C_1 \ar[rr]_-{\cmp}&& \C_1
}
\end{gathered}
\end{flalign}
\begin{flalign}
\begin{gathered}
\xymatrix@R=1.5em@C=0.5em{
&&\ar[dll]_-{(ut)\times_{\C_0}\id}\ar[dd]|{_{~}\id^{~}}\C_1 \ar[drr]^-{\id\times_{\C_0}(us)}&&\\
\ar[drr]_-{\cmp}\C_1\times_{\C_0}\C_1 &\stackrel{\mathsf{l}}{\Rightarrow}&~&\stackrel{\mathsf{r}}{\Leftarrow}& \C_1\times_{\C_0}\C_1\ar[dll]^-{\cmp}\\
&&\C_1&&
}
\end{gathered}
\end{flalign}
\end{subequations}
\end{itemize}
The natural isomorphisms $(\mathsf{a},\mathsf{l},\mathsf{r})$ are required 
to be globular, i.e.\ the images of their components under $s$ and $t$ are identities in $\C_0$,
and to satisfy the typical unity and pentagon axioms, see e.g.\ \cite[Definition 12.3.7]{Yau}.
\end{defi}
\begin{rem}
A pseudo-category $\C$ consists of \textit{objects}
$c\in \C_0$, \textit{vertical morphisms} $(g : c\to c^\prime)\in\C_0$,
\textit{horizontal morphisms} $f\in\C_1$ and \textit{$2$-cells} $(\alpha : f\Rightarrow f^\prime)\in\C_1$.
Using the functors $s$ and $t$, one can assign a source and a target in $\C_0$ to each horizontal morphism
and to each $2$-cell. To distinguish between vertical and horizontal morphisms, we shall denote
the latter by a slashed arrow $f : c_0 \hto c_1$, where $c_0=s(f)$ and $c_1 = t(f)$. The
$2$-cells will be denoted by squares
\begin{flalign}
\begin{gathered}
\xymatrix@R=0.25em@C=0.25em{
c_0^\prime  \ar[rr]|{\sla}^-{f^\prime}&& c_1^\prime\\
&{\scriptstyle\alpha}~\rotatebox[origin=c]{90}{$\Rightarrow$}~~&\\
\ar[uu]^-{g_0} c_0 \ar[rr]|{\sla}_-{f} && c_1 \ar[uu]_-{g_1}
}
\end{gathered}\qquad,
\end{flalign}
where $g_0 = s(\alpha)$ and $g_1 = t(\alpha)$. The compositions
in the groupoids $\C_0$ and $\C_1$ define, respectively, a vertical composition
$g^\prime\, g: c\to c^{\prime\prime}$ of vertical morphisms 
$g:c\to c^\prime$  and $g^\prime:c^\prime \to c^{\prime\prime}$,
and a vertical composition $\alpha^\prime\,\alpha : f\Rightarrow f^{\prime\prime}$ of $2$-cells
$\alpha : f\Rightarrow f^\prime$ and $\alpha^\prime : f^\prime\Rightarrow f^{\prime\prime}$. 
These vertical compositions are strictly associative and unital with respect to the 
identities $(\id_{c}:c\to c)\in \C_0$ and $(\id_f : f\Rightarrow f)\in \C_1$.
The functor $\cmp$ defines a horizontal composition $f_1\cmp f_0$ of horizontal morphisms
$f_0 : c_0 \hto c_1$ and $f_1 : c_1\hto c_2$, and a horizontal composition of $2$-cells
\begin{flalign}
\begin{gathered}
\xymatrix@R=0.25em@C=0.25em{
c_0^\prime  \ar[rr]|{\sla}^-{f_0^\prime}&& c_1^\prime \ar[rr]|{\sla}^-{f_1^\prime} && c_2^\prime\\
&{\scriptstyle\alpha_0}~\rotatebox[origin=c]{90}{$\Rightarrow$}~~& &{\scriptstyle\alpha_1}~\rotatebox[origin=c]{90}{$\Rightarrow$}~~&\\
\ar[uu]^-{g_0} c_0 \ar[rr]|{\sla}_-{f_0} && c_1 \ar[uu]^-{g_1} \ar[rr]|{\sla}_-{f_1}&& c_2\ar[uu]_-{g_2}
}
\end{gathered}
~~~\stackrel{\cmp}{\longmapsto}~~~
\begin{gathered}
\xymatrix@R=0.25em@C=0.25em{
c_0^\prime  \ar[rr]|{\sla}^-{f_1^\prime\cmp f_0^\prime}&& c_2^\prime \\
&{\scriptstyle\alpha_1\cmp \alpha_0}~\rotatebox[origin=c]{90}{$\Rightarrow$}~~~~&\\
\ar[uu]^-{g_0} c_0 \ar[rr]|{\sla}_-{f_1\cmp f_0} && c_2\ar[uu]_-{g_2}
}
\end{gathered}\qquad.
\end{flalign}
These horizontal compositions are only weakly associative, with associator $\mathsf{a}$,
and weakly unital, with unitors $\mathsf{l}$ and $\mathsf{r}$,
with respect to the units obtained by the functor $u:\C_0\to \C_1$.
Note that the two compositions of $2$-cells satisfy the strict interchange law
\begin{flalign}
(\alpha_1^\prime\,\alpha_1)\cmp (\alpha_0^\prime\,\alpha_0) \,=\,
(\alpha_1^\prime \cmp \alpha_0^\prime)~(\alpha_1\cmp\alpha_0)\quad,
\end{flalign}
as a consequence of the functoriality of $\cmp$.
\end{rem}

\begin{defi}\label{def:PsFunctor}
A \textit{pseudo-functor} $F : \C\to\D$ between two pseudo-categories
is a tuple
\begin{flalign}
F \,=\, (F_0,F_1,F^{\cmp},F^u) 
\end{flalign}
that consists of the following data:
\begin{itemize}
\item[(i)] Two functors $F_0 : \C_0\to\D_0$ and $F_1 : \C_1\to \D_1$ inducing
a map of spans, i.e.\ $s_{\D}\,F_1 = F_0\,s_{\C}$ and $t_{\D}\,F_1 = F_0\,t_{\C}$.

\item[(ii)] Natural isomorphisms $F^{\cmp} : \cmp_\D \, (F_1\times F_1) \Rightarrow F_1\, \cmp_\C$
and $F^u : u_\D\,F_0\Rightarrow F_1\,u_\C$.
\end{itemize}
The natural isomorphisms $F^{\cmp}$ and $F^u$ are required to be globular, i.e.\ 
the images of their components under $s_\D$ and $t_\D$ are identities,
and to satisfy analogous coherence conditions to those of a monoidal functor, 
see e.g.\ \cite[Definition 12.3.18]{Yau}.
\end{defi}

\begin{defi}\label{def:Transformation}
A \textit{(vertical) transformation} $\zeta : F\Rightarrow G$ between two
pseudo-functors $F,G : \C\to\D$ consists of two natural transformations
$\zeta^0 : F_0\Rightarrow G_0$ and $\zeta^1 : F_1\Rightarrow G_1$ which satisfy
the following conditions:
\begin{itemize}
\item[(1)] $s_{\D}\,\zeta^1 = \zeta^0\,s_{\C}$ and $t_{\D}\,\zeta^1 = \zeta^0\,t_{\C}$.

\item[(2)] For all horizontal morphisms $f_0 : c_0\hto c_1$ and $f_1 : c_1\hto c_2$,
the compositions of $2$-cells
\begin{flalign}
\begin{gathered}
\xymatrix@R=0.25em@C=0.75em{
G_0(c_0) \ar[rr]|{\sla}^-{G_1(f_1\cmp f_0)} &&G_0(c_2)\\
&{\scriptstyle\zeta^1_{f_1\cmp f_0}}~\rotatebox[origin=c]{90}{$\Rightarrow$}~&\\
\ar[uu]^-{\zeta^0_{c_0}} F_0(c_0)  \ar[rr]|{\sla}^-{F_1(f_1\cmp f_0)}&& F_0(c_2) \ar[uu]_-{\zeta^0_{c_2}}\\
& {\scriptstyle F^{\cmp}_{(f_1,f_0)}}~\rotatebox[origin=c]{90}{$\Rightarrow$}~&\\
\ar@{=}[uu] F_0(c_0) \ar[r]|{\sla}_-{F_1(f_0)} & F_0(c_1) \ar[r]|{\sla}_-{F_1(f_1)}& 
F_0(c_2)\ar@{=}[uu]
}
\end{gathered}
~~=~~
\begin{gathered}
\xymatrix@R=0.25em@C=0.25em{
G_0(c_0) \ar[rrrr]|{\sla}^-{G_1(f_1\cmp f_0)} && && G_0(c_2)\\
&& {\scriptstyle G^{\cmp}_{(f_1,f_0)}}~\rotatebox[origin=c]{90}{$\Rightarrow$}~&&\\
\ar@{=}[uu]G_0(c_0) \ar[rr]|{\sla}^-{G_1(f_0)} && G_0(c_1) \ar[rr]|{\sla}^-{G_1(f_1)}& &
G_0(c_2)\ar@{=}[uu]\\
&{\scriptstyle \zeta^1_{f_0}}~\rotatebox[origin=c]{90}{$\Rightarrow$}&  &{\scriptstyle \zeta^1_{f_1}}~\rotatebox[origin=c]{90}{$\Rightarrow$}&\\
\ar[uu]^-{\zeta^0_{c_0}}F_0(c_0) \ar[rr]|{\sla}_-{F_1(f_0)} && \ar[uu]^-{\zeta^0_{c_1}}F_0(c_1) \ar[rr]|{\sla}_-{F_1(f_1)}& & F_0(c_2)\ar[uu]_-{\zeta^0_{c_2}}
}
\end{gathered}
\end{flalign}
coincide.

\item[(3)] For all objects $c\in\C_0$, the compositions of $2$-cells
\begin{flalign}
\begin{gathered}
\xymatrix@R=0.25em@C=0.25em{
G_0(c)\ar[rr]|{\sla}^-{G_1 u(c)}&&G_0(c)\\
&{\scriptstyle \zeta^1_{u(c)}}~\rotatebox[origin=c]{90}{$\Rightarrow$}&\\
\ar[uu]^-{\zeta^0_c} F_0(c)\ar[rr]|{\sla}^-{F_1u(c)}&& F_0(c)\ar[uu]_-{\zeta^0_c}\\
&{\scriptstyle F^u_{c}}~\rotatebox[origin=c]{90}{$\Rightarrow$}&\\
\ar@{=}[uu] F_0(c) \ar[rr]|{\sla}_-{uF_0(c)}&& F_0(c)\ar@{=}[uu] 
}
\end{gathered}
~~=~~
\begin{gathered}
\xymatrix@R=0.25em@C=0.25em{
G_0(c)  \ar[rr]|{\sla}^-{G_1 u(c)}&&G_0(c)\\
&{\scriptstyle G^u_{c}}~\rotatebox[origin=c]{90}{$\Rightarrow$}&\\
\ar@{=}[uu] G_0(c)  \ar[rr]|{\sla}^-{uG_0(c)}&&G_0(c) \ar@{=}[uu] \\
&{\scriptstyle u(\zeta^0_{c})}~\rotatebox[origin=c]{90}{$\Rightarrow$}&\\
\ar[uu]^-{\zeta^0_c} F_0(c) \ar[rr]|{\sla}_-{uF_0(c)}&& F_0(c)\ar[uu]_-{\zeta^0_c}
}
\end{gathered}
\end{flalign}
coincide.
\end{itemize}
\end{defi}

The following result is shown in \cite{PseudoCats}.
We also refer to \cite[Chapter 12.3]{Yau} for a sufficiently detailed sketch.
\begin{propo}\label{prop:DblCat}
There is a (strict) $2$-category $\PsCat$ whose objects
are pseudo-categories, morphisms are pseudo-functors and $2$-morphisms
are (vertical) transformations.
\end{propo}

\begin{rem}\label{rem:PsCat2,1}
The $2$-category $\PsCat$ is in fact a $(2,1)$-category.
Indeed, since the vertical morphisms and $2$-cells are invertible in 
any pseudo-category as in Definition~\ref{def:PsCat}, 
each transformation as in Definition~\ref{def:Transformation} is invertible.
\end{rem}

Of interest to us is the full $2$-subcategory $\PsCat^{\mathrm{fib}}\subseteq \PsCat$
of \textit{fibrant} pseudo-categories.
The relevant definition from \cite[Definition 3.4]{Shulman} 
simplifies in our case because pseudo-categories are special 
double categories in which every vertical morphism is invertible, which by \cite[Lemma 3.20]{Shulman}
implies that we do not have to introduce the concept of conjoints.
\begin{defi}\label{def:fibrant}
Let $\C\in\PsCat$ be a pseudo-category.
\begin{itemize} 
\item[(a)] A \textit{companion} of a vertical morphism
$g : c_0\to c_1$ is a horizontal morphism $\hat{g}: c_0\hto c_1$ together with $2$-cells
\begin{flalign}
\begin{gathered}
\xymatrix@R=0.25em@C=0.25em{
c_1 \ar[rr]|{\sla}^-{u(c_1)} && c_1\\
&\rotatebox[origin=c]{90}{$\Rightarrow$}&\\
\ar[uu]^-{g}c_0 \ar[rr]|{\sla}_-{\hat{g}} && c_1 \ar@{=}[uu]
}
\end{gathered}
\qquad\quad\text{and}\qquad\quad
\begin{gathered}
\xymatrix@R=0.25em@C=0.25em{\\
c_0 \ar[rr]|{\sla}^-{\hat{g}}&& c_1\\
&\rotatebox[origin=c]{90}{$\Rightarrow$}&\\
\ar@{=}[uu]c_0 \ar[rr]|{\sla}_-{u(c_0)}&& c_0\ar[uu]_-{g}
}
\end{gathered}\qquad,
\end{flalign}
such that
\begin{subequations}\label{eqn:companionidentities}
\begin{flalign}\label{eqn:companionidentities1}
\begin{gathered}
\xymatrix@R=0.25em@C=0.25em{
c_1 \ar[rr]|{\sla}^-{u(c_1)} && c_1\\
&\rotatebox[origin=c]{90}{$\Rightarrow$}&\\
\ar[uu]^-{g}c_0 \ar[rr]|{\sla}^-{\hat{g}}&& c_1\ar@{=}[uu]\\
&\rotatebox[origin=c]{90}{$\Rightarrow$}&\\
\ar@{=}[uu]c_0 \ar[rr]|{\sla}_-{u(c_0)}&& c_0\ar[uu]_-{g}
}
\end{gathered}
~~=~~
\begin{gathered}
\xymatrix@R=0.25em@C=0.25em{
c_1 \ar[rr]|{\sla}^-{u(c_1)} && c_1\\
&{\scriptstyle u(g)}~\rotatebox[origin=c]{90}{$\Rightarrow$}&\\
\ar[uu]^-{g} c_0 \ar[rr]|{\sla}_-{u(c_0)}&& c_0\ar[uu]_-{g}
}
\end{gathered}
\end{flalign}
and
\begin{flalign}\label{eqn:companionidentities2}
\begin{gathered}
\xymatrix@R=0.25em@C=0.25em{
c_0 \ar[rrrr]|{\sla}^-{\hat{g}} && &&c_1\\
&& {\scriptstyle \mathsf{l}_{\hat{g}}}~\rotatebox[origin=c]{90}{$\Rightarrow$}~{\scriptstyle\cong} && \\
\ar@{=}[uu]c_0 \ar[rr]|{\sla}^-{\hat{g}}&& c_1 \ar[rr]|{\sla}^-{u(c_1)} && c_1\ar@{=}[uu]\\
&\rotatebox[origin=c]{90}{$\Rightarrow$}&   &\rotatebox[origin=c]{90}{$\Rightarrow$}&  \\
\ar@{=}[uu]c_0 \ar[rr]|{\sla}_-{u(c_0)}&& c_0 \ar[uu]^-{g} \ar[rr]|{\sla}_-{\hat{g}} && c_1\ar@{=}[uu]\\
&& {\scriptstyle \mathsf{r}_{\hat{g}}}~\rotatebox[origin=c]{-90}{$\Rightarrow$}~{\scriptstyle\cong} && \\
\ar@{=}[uu]c_0 \ar[rrrr]|{\sla}_-{\hat{g}}&& &&c_1\ar@{=}[uu]
}
\end{gathered}
~~=~~
\begin{gathered}
\xymatrix@R=0.25em@C=0.25em{
c_0 \ar[rr]|{\sla}^-{\hat{g}}&& c_1\\
& {\scriptstyle \id_{\hat{g}}}~\rotatebox[origin=c]{90}{$\Rightarrow$}&\\
\ar@{=}[uu]c_0 \ar[rr]|{\sla}_-{\hat{g}}&& c_1\ar@{=}[uu]
}
\end{gathered}\qquad.
\end{flalign}
\end{subequations}

\item[(b)] The pseudo-category $\C$ is called \textit{fibrant}
if every vertical morphism has a companion.
We denote by $\PsCat^{\mathrm{fib}}\subseteq\PsCat$ the full $2$-subcategory
of fibrant pseudo-categories. 
\end{itemize}
\end{defi}

\begin{rem}
In \cite[Section 3]{Shulman}, Shulman provides a list of
technical lemmas for companions. These enter frequently in 
proving our more technical results in Appendix \ref{app:pseudocats}.
\end{rem}

In the main part of this paper, we require constructions
that allow us to assign to an ordinary category $\CC\in\Cat$
a pseudo-category $\iota(\CC)\in\PsCat$ and to a pseudo-category
$\C\in\PsCat$ a suitably truncated ordinary category $\tau(\C)\in\Cat$. The latter 
can be viewed as the \textit{homotopy category} of the pseudo-category $\C$.
These constructions are developed in detail in Appendix \ref{app:pseudocats},
where we also prove the following result.
\begin{theo}\label{theo:2adjunction}
There exists a $2$-adjunction 
\begin{flalign}\label{eqn:tauiotaadjunction}
\xymatrix@C=3.5em{
\tau ~:~ \PsCat^{\mathrm{fib}} \ar@<1.2ex>[r]_-{\perp}~~&~~\ar@<1.2ex>[l] \Cat^{(2,1)} ~:~ \iota
}
\end{flalign}
between the $(2,1)$-category $\PsCat^{\mathrm{fib}}$ of fibrant pseudo-categories
and the $(2,1)$-category $\Cat^{(2,1)}$ of categories, functors and natural \textit{isomorphisms}.
In particular, given any fibrant pseudo-category $\C\in\PsCat^{\mathrm{fib}}$ and any ordinary category
$\DD\in\Cat^{(2,1)}$, there exists an equivalence (in fact, an isomorphism)
\begin{flalign}
\PsFun\big(\C,\iota(\DD)\big) ~\cong~ \Fun\big(\tau(\C),\DD \big)
\end{flalign}
between the groupoid $\PsFun\big(\C,\iota(\DD)\big)$ of pseudo-functors
from $\C$ to $\iota(\DD)$ and their transformations, and the groupoid 
$\Fun\big(\tau(\C),\DD \big)$ of ordinary functors from $\tau(\CC)$ to $\DD$
and their natural isomorphisms. These equivalences are natural in $\C$ and $\DD$.
\end{theo}

\begin{rem}
The appearance of the $(2,1)$-category $\Cat^{(2,1)}$, in contrast to the 
$2$-category $\Cat$ of categories, functors and (not necessarily invertible)
natural transformations, is due to the fact that every (vertical)
transformation in $\PsCat$ is invertible, see Remark \ref{rem:PsCat2,1}. At the 
level of pseudo-categories, there exists a more general concept of transformations,
called horizontal or pseudo-natural \cite[Section 3]{PseudoCats}, which are not necessarily 
invertible. We expect that Theorem~\ref{theo:2adjunction} can be upgraded to 
a bicategorical adjunction between the $2$-category $\Cat$ and the 
bicategory of fibrant pseudo-categories, pseudo-functors and \textit{horizontal} transformations.
This potential generalization is, however, not needed in the present paper,
because we are interested in describing \textit{groupoids} of quantum field theories,
and the invertible horizontal transformations between pseudo-functors $\C\to\iota(\DD)$
coincide with the vertical transformations from Definition~\ref{def:Transformation}.
\end{rem}

In the remainder of this section we describe the action of the 
$2$-functors $\iota$ and $\tau$ on objects, which should allow the reader
to understand the main concepts behind these $2$-functors 
without being confronted with the more technical aspects of Appendix~\ref{app:pseudocats}.
\begin{constr}\label{constr:iota}
Given any ordinary category $\CC\in\Cat$, we define the
pseudo-category $\iota(\CC)\in\PsCat$ by the following data (compare Definition~\ref{def:PsCat}):
\begin{itemize}
\item[(i)] $\iota(\CC)_0 =\mathrm{core}(\CC)\subseteq \CC$ is the core of the given category, i.e.\
the maximal subgroupoid consisting of all objects and all isomorphisms 
$g:c\stackrel{\cong}{\longrightarrow} c^\prime$ in $\CC$.

\item[(ii)] $\iota(\CC)_1$ is the groupoid whose objects are morphisms $(f:c_0\to c_1)\in \CC$ 
and whose morphisms are commutative squares
\begin{flalign}
\begin{gathered}
\xymatrix{
c_0^\prime \ar[r]^-{f^\prime} & c_1^\prime\\
\ar[u]_-{\cong}^-{g_0}c_0\ar[r]_-{f}&c_1\ar[u]^-{\cong}_-{g_1}
}
\end{gathered}
\end{flalign}
in $\CC$ with both vertically drawn morphisms being isomorphisms. 
(To ease notation, we will often suppress the symbols $\cong$.)
The source functor $s$ sends such square to the isomorphism $g_0 : c_0\to c_0^\prime$ in $\CC$,
and the target functor $t$ sends it to the isomorphism $g_1 : c_1\to c_1^\prime$.

\item[(iii)] The horizontal composition functor $\cmp$ is defined by composition in $\CC$ as
\begin{flalign}
\begin{gathered}
\xymatrix{
c_0^\prime \ar[r]^-{f_0^\prime} & c_1^\prime \ar[r]^-{f_1^\prime} & c_2^\prime\\
\ar[u]^-{g_0}c_0\ar[r]_-{f_0}&c_1\ar[u]^-{g_1} \ar[r]_-{f_1}& c_2\ar[u]_-{g_2}
}
\end{gathered}
~~\stackrel{\cmp}{\longmapsto}~~
\begin{gathered}
\xymatrix{
c_0^\prime \ar[r]^-{f_1^\prime\,f_0^\prime} & c_2^\prime\\
\ar[u]^-{g_0}c_0\ar[r]_-{f_1\, f_0} & c_2\ar[u]_-{g_2}
}
\end{gathered} \quad .
\end{flalign}

\item[(iv)] The horizontal unit functor $u$ is defined by the identities in $\CC$ via
\begin{flalign}
\begin{gathered}
\xymatrix{
c^\prime\\
\ar[u]^-{g}c
}
\end{gathered}
~~\stackrel{u}{\longmapsto}~~
\begin{gathered}
\xymatrix{
c^\prime \ar[r]^-{\id_{c^\prime}}& c^\prime\\
\ar[u]^-{g}c \ar[r]_-{\id_c}& c\ar[u]_-{g}
}
\end{gathered} \quad .
\end{flalign}

\item[(v)] The associator and unitors are trivial, i.e.\ they consist of the identity natural transformations.
\end{itemize}
Each pseudo-category of the form $\iota(\CC)$ is fibrant. Indeed, a companion of a vertical morphism 
$g : c \stackrel{\cong}{\longrightarrow} c^\prime$ is the same morphism $\hat{g} = g$ drawn horizontally.
\end{constr}

\begin{constr}\label{constr:tau}
Given any pseudo-category $\C\in\PsCat$, we define an ordinary
category $\tau(\C)\in\Cat$ by the following objects and morphisms:
\begin{itemize}
\item[\underline{$\mathsf{Obj}$:}] Objects $c\in\C_0$ of the pseudo-category.

\item[\underline{$\mathsf{Mor}$:}] Equivalence classes $[f : c_0\hto c_1]$ of 
horizontal morphisms for the following equivalence relation: Two horizontal
morphisms $f : c_0\hto c_1$ and $f^\prime : c_0\hto c_1$ with the same source and target 
are equivalent if there exists a globular $2$-cell
\begin{flalign}
\begin{gathered}
\xymatrix@R=0.25em@C=0.25em{
c_0 \ar[rr]|{\sla}^-{f^\prime}&& c_1\\
&{\scriptstyle \cong }~\rotatebox[origin=c]{90}{$\Rightarrow$}&\\
\ar@{=}[uu]c_0 \ar[rr]|{\sla}_-{f}&& c_1\ar@{=}[uu]
}
\end{gathered}\qquad,
\end{flalign}
which is automatically an isomorphism because $\C_1$ is a groupoid.
\end{itemize}
Composition of morphisms in the category $\tau(\C)$
is defined by horizontal composition $[f_1]\,[f_0]:= [f_1\cmp f_0]$ of  any choice of 
representatives and the identities are defined by the horizontal unit $[u(c)]$. Associativity
and unitality hold strictly because the natural isomorphisms $(\mathsf{a},\mathsf{l},\mathsf{r})$
from Definition~\ref{def:PsCat} are globular, hence they are trivial at the level of equivalence classes.
Note that the definition of $\tau(\C)$ does not require that $\C$ is fibrant. The fibrancy condition will become
crucial when defining the $2$-functor $\tau$ on $2$-morphisms,
see Appendix \ref{app:pseudocats}.
\end{constr}


\section{\label{sec:LBord}The globally hyperbolic Lorentzian bordism pseudo-category}
In this section we define an analogue of the geometric bordism pseudo-categories
of Stolz and Teichner \cite{StolzTeichner} for globally hyperbolic Lorentzian manifolds.
This allows us to introduce a concept of functorial field theories in the Lorentzian context,
and to compare the latter with algebraic quantum field theory.
Before spelling out the details of our bordism pseudo-categories, we would like to add
two clarifying comments:

\begin{rem}
In analogy to \cite{StolzTeichner}, our bordisms do \textit{not} have boundaries, but 
they have marked codimension $1$ hypersurfaces together with collar neighborhoods.
These collars are needed to obtain a well-defined gluing construction when 
bordisms are endowed with geometric structures.
In our case the geometric structures of central interest consist of Lorentzian metrics. 
Controlling these collar neighborhoods is precisely the reason why we have to work
with pseudo-categories instead of ordinary categories.
\end{rem}

\begin{rem}\label{rem:GlobHyp and topology}
We are interested in oriented and time-oriented \textit{globally hyperbolic} Lorentzian manifolds, which by 
Definition~\ref{def:globhyp} come with a distinguished class of surfaces, the Cauchy surfaces.
The bordisms we consider in this work are thus described by objects $N\in\Loc_m$
in the category introduced in Definition~\ref{def:Loc},
and they go from a Cauchy surface $\Sigma_0\subset N$ to another one $\Sigma_1\subset N$
that lies in the causal future of $\Sigma_0$. In physics 
terminology, our bordisms admit an interpretation in terms of \textit{time evolution}
from a Cauchy surface to another one.
By a fundamental result in Lorentzian geometry 
\cite{BernalSanchez,BernalSanchez2}, each of these globally hyperbolic bordisms has an underlying
manifold $N\cong \bbR\times \Sigma$ that is diffeomorphic to a cylinder.  This implies that
our globally hyperbolic bordism pseudo-category only describes bordisms whose underlying
manifolds are cylinders. It is important to stress that such types of bordisms are still 
highly non-trivial, because, in general, each cylinder $\bbR\times \Sigma$ has
a rich moduli space of globally hyperbolic Lorentzian metrics.
\end{rem}

We will now define a pseudo-category $\LB_m\in\PsCat$ of $(m\geq 1)$-dimensional
oriented and time-oriented globally hyperbolic Lorentzian bordisms. For this we list
and explain all the required data from Definition~\ref{def:PsCat}.

\paragraph{The groupoid $(\LB_m)_0$:} The groupoid $(\LB_m)_0$ 
of objects of the pseudo-category $\LB_m\in\PsCat$ is defined by the following
objects and morphisms:
\begin{itemize}
\item[\underline{$\mathsf{Obj}$:}] Pairs $(M,\Sigma)$ consisting of an object $M\in\Loc_m$
and a Cauchy surface $\Sigma\subset M$. 
\sk

Geometrically, one interprets such pairs as
Cauchy surfaces $\Sigma$ with collar regions given by $M$, i.e.\
\begin{flalign}
\begin{gathered}
\begin{tikzpicture}[scale=1]
\draw[dashed,fill=gray!5] (-0.8,-1) -- (-0.8,1) -- (0.8,1) -- (0.8,-1) -- (-0.8,-1);
\draw (0,0.8) node {{\footnotesize $M$}}; 
\draw[thick] (-0.8,0) .. controls (-0.4,0.1) .. (0,0) .. controls (0.4,-0.1) .. (0.8,0) 
node[pos=0.2, above] {{\footnotesize $\Sigma$}};
\end{tikzpicture}
\end{gathered}\qquad.
\end{flalign}

\item[\underline{$\mathsf{Mor}$:}] Equivalence classes 
$[W,g]: (M,\Sigma)\to (M^\prime,\Sigma^\prime)$ of pairs $(W,g)$ consisting of a causally 
convex open subset $W\subseteq M$ that contains the Cauchy surface $\Sigma\subset W$ 
and of a Cauchy morphism $g:W\to M^\prime$ in $\Loc_m $ satisfying $g(\Sigma) = \Sigma^\prime$.
Two such pairs $(W,g)$ and $(\widetilde{W},\widetilde{g})$ are equivalent if and only if there exists
a causally convex open subset $\widehat{W}\subseteq W\cap \widetilde{W} \subseteq M$ 
which contains the Cauchy surface $\Sigma\subset \widehat{W}$, such that the restrictions
$g\vert_{\widehat{W}} = \widetilde{g}\vert_{\widehat{W}}$ coincide.
\sk

It is convenient to visualize a representative $(W,g)$ of the equivalence 
class $[W,g] : (M,\Sigma)\to (M^\prime,\Sigma^\prime)$ defining a morphism
by a zig-zag
\begin{flalign}
\xymatrix{
M &\ar[l]_-{\subseteq} W \ar[r]^-{g}& M^\prime
}
\end{flalign}
of Cauchy morphisms in $\Loc_m$, where it is always implicitly understood 
that each object contains the relevant marked Cauchy surface and that each map
preserves these Cauchy surfaces. Geometrically,
one interprets $[W,g]$ as the germ of a local Cauchy morphism
\begin{flalign}
\begin{gathered}
\begin{tikzpicture}[scale=1]
\draw[dashed,fill=gray!5] (-0.8,-1) -- (-0.8,1) -- (0.8,1) -- (0.8,-1) -- (-0.8,-1);
\draw[dashed,fill=gray!20] (-0.8,-0.5) -- (-0.8,0.5) -- (0.8,0.5) -- (0.8,-0.5) -- (-0.8,-0.5);
\draw (0,0.8) node {{\footnotesize $M$}}; 
\draw[thick] (-0.8,0) .. controls (-0.4,0.1) .. (0,0) .. controls (0.4,-0.1) .. (0.8,0) 
node[pos=0.2, above] {{\footnotesize $\Sigma$}};
\draw[->,thick] (1.8,0) -- (1.2,0) node [midway, above] {{\footnotesize $\subseteq$}};
\draw[dashed,fill=gray!20] (2.2,-0.5) -- (2.2,0.5) -- (3.8,0.5) -- (3.8,-0.5) -- (2.2,-0.5);
\draw (3,-0.3) node {{\footnotesize $W$}}; 
\draw[thick] (2.2,0) .. controls (2.6,0.1) .. (3,0) .. controls (3.4,-0.1) .. (3.8,0) 
node[pos=0.2, above] {{\footnotesize $\Sigma$}};
\draw[->,thick] (4.2,0.1) -- (4.8,0.4) node [midway, above] {{\footnotesize $g$}};
\draw[dashed,fill=gray!5] (5.2,-1.5) -- (5.2,1.5) -- (6.8,1.5) -- (6.8,-1.5) -- (5.2,-1.5);
\draw[dashed,fill=gray!20] (5.2,0) -- (5.2,1) -- (6.8,1) -- (6.8,0) -- (5.2,0);
\draw (6,-0.75) node {{\footnotesize $M^\prime$}}; 
\draw[thick] (5.2,0.5) .. controls (5.6,0.6) .. (6,0.5) .. controls (6.4,0.4) .. (6.8,0.5) 
node[pos=0.2, above] {{\footnotesize $\Sigma^\prime$}};
\end{tikzpicture}
\end{gathered}\qquad
\end{flalign}
that identifies locally isomorphic collar regions around Cauchy surfaces.
\end{itemize}

The identity morphisms
$[M,\id]: (M,\Sigma)\to(M,\Sigma)$ of the groupoid $(\LB_m)_0$ are
defined as the equivalence classes of the zig-zags
\begin{flalign}
\xymatrix{
M &\ar[l]_-{=} M \ar[r]^-{\id}& M
}\quad.
\end{flalign}
To define the composition of two morphisms 
$[W,g] : (M,\Sigma)\to (M^\prime,\Sigma^\prime)$ and
$[W^\prime,g^\prime] : (M^\prime,\Sigma^\prime)\to (M^{\prime\prime},\Sigma^{\prime\prime})$, 
we choose any representatives, factorize the corresponding chain of zig-zags as
\begin{subequations}
\begin{flalign}
\begin{gathered}
\xymatrix{
M &\ar[l]_-{\subseteq} W \ar[r]^-{g}& M^\prime & \ar[l]_-{\subseteq} W^\prime \ar[r]^-{g^\prime}& M^{\prime\prime}\\
&\ar[u]^-{\subseteq}g^{-1}(W^\prime) \ar@{-->}[rru]_-{g}&&&
}
\end{gathered}
\end{flalign}
and define the composite morphism by
\begin{flalign}
[W^\prime,g^\prime]\,[W,g]\,:=\, \big[g^{-1}(W^\prime),g^\prime\,g\big]\,:\, (M,\Sigma)~\longrightarrow~(M^{\prime\prime},\Sigma^{\prime\prime})\quad,
\end{flalign}
\end{subequations}
where $g^{-1}(W^\prime)\subseteq M$ denotes the preimage of $W^\prime$ under $g$. 
(Note that the subset $g^{-1}(W^\prime) = g^{-1}\big(g(M)\cap W^{\prime}\big)\subseteq M$
is causally convex by Remarks \ref{rem:causallyconvex} and \ref{rem:factorization}.) 
One checks that every morphism $[W,g]:(M,\Sigma)\to(M^\prime,\Sigma^\prime)$
in $(\LB_m)_0$ is an isomorphism, with inverse given explicitly by
$[g(W),g^{-1}] : (M^\prime,\Sigma^\prime)\to (M,\Sigma)$.

\paragraph{The groupoid $(\LB_m)_1$:} The groupoid $(\LB_m)_1$ 
of morphisms of the pseudo-category $\LB_m\in\PsCat$ is defined by the following
objects and morphisms:
\begin{itemize}
\item[\underline{$\mathsf{Obj}$:}] Tuples $(N,i_0,i_1) : (M_0,\Sigma_0)\hto (M_1,\Sigma_1)$ 
consisting of an object $N\in\Loc_m$ and zig-zags
\begin{subequations}\label{eqn:LBord1cell}
\begin{flalign}
\xymatrix{
M_0 &\ar[l]_-{\subseteq} V_0 \ar[r]^-{i_0}& N & \ar[l]_-{i_1} V_1 \ar[r]^-{\subseteq}& M_1 
}
\end{flalign}
of Cauchy morphisms in $\Loc_m$ satisfying 
\begin{flalign}
\Sigma_0\subset V_0\subseteq M_0~~,\quad
\Sigma_1\subset V_1\subseteq M_1~~,\quad
i_1(\Sigma_1)\subset J^+_N\big(i_0(\Sigma_0)\big)\quad.
\end{flalign} 
\end{subequations}
The last condition states that, when embedded into $N$, 
the Cauchy surface $\Sigma_1$ lies in the causal future of $\Sigma_0$.
\sk

Geometrically, one interprets such a datum as a bordism
\begin{flalign}
\begin{gathered}\begin{tikzpicture}[scale=1]
\draw[dashed,fill=gray!5] (-0.8,-1) -- (-0.8,1) -- (0.8,1) -- (0.8,-1) -- (-0.8,-1);
\draw[dashed,fill=gray!20] (-0.8,-0.5) -- (-0.8,0.5) -- (0.8,0.5) -- (0.8,-0.5) -- (-0.8,-0.5);
\draw (0,0.8) node {{\footnotesize $M_0$}}; 
\draw[thick] (-0.8,0) .. controls (-0.4,0.1) .. (0,0) .. controls (0.4,-0.1) .. (0.8,0) 
node[pos=0.2, above] {{\footnotesize $\Sigma_0$}};
\draw[->,thick] (1.8,0) -- (1.2,0) node [midway, above] {{\footnotesize $\subseteq$}};
\draw[dashed,fill=gray!20] (2.2,-0.5) -- (2.2,0.5) -- (3.8,0.5) -- (3.8,-0.5) -- (2.2,-0.5);
\draw (3,-0.3) node {{\footnotesize $V_0$}}; 
\draw[thick] (2.2,0) .. controls (2.6,0.1) .. (3,0) .. controls (3.4,-0.1) .. (3.8,0) 
node[pos=0.2, above] {{\footnotesize $\Sigma_0$}};
\draw[->,thick] (4.2,0) -- (4.8,0) node [midway, above] {{\footnotesize $i_0$}};
\draw[dashed,fill=gray!5] (5.2,-1.5) -- (5.2,3.5) -- (6.8,3.5) -- (6.8,-1.5) -- (5.2,-1.5);
\draw[dashed,fill=gray!20] (5.2,-0.5) -- (5.2,0.5) -- (6.8,0.5) -- (6.8,-0.5) -- (5.2,-0.5);
\draw[dashed,fill=gray!20] (5.2,1.5) -- (5.2,2.5) -- (6.8,2.5) -- (6.8,1.5) -- (5.2,1.5);
\draw (6,3) node {{\footnotesize $N$}}; 
\draw[thick] (5.2,0) .. controls (5.6,0.1) .. (6,0) .. controls (6.4,-0.1) .. (6.8,0.1) 
node[pos=0.2, above] {{\footnotesize $i_0(\Sigma_0)~~$}};
\draw[thick] (5.2,2) .. controls (5.6,2.1) .. (6,2) .. controls (6.4,1.9) .. (6.8,2) 
node[pos=0.2, above] {{\footnotesize $i_1(\Sigma_1)~~$}};
\draw[->,thick] (7.8,2) -- (7.2,2) node [midway, above] {{\footnotesize $i_1$}};
\draw[dashed,fill=gray!20] (8.2,1.5) -- (8.2,2.5) -- (9.8,2.5) -- (9.8,1.5) -- (8.2,1.5);
\draw (9,1.7) node {{\footnotesize $V_1$}}; 
\draw[thick] (8.2,2) .. controls (8.6,2.1) .. (9,2) .. controls (9.4,1.9) .. (9.8,2) 
node[pos=0.2, above] {{\footnotesize $\Sigma_1$}};
\draw[->,thick] (10.2,2) -- (10.8,2) node [midway, above] {{\footnotesize $\subseteq$}};
\draw[dashed,fill=gray!5] (11.2,1) -- (11.2,3) -- (12.8,3) -- (12.8,1) -- (11.2,1);
\draw[dashed,fill=gray!20] (11.2,1.5) -- (11.2,2.5) -- (12.8,2.5) -- (12.8,1.5) -- (11.2,1.5);
\draw (12,2.8) node {{\footnotesize $M_1$}}; 
\draw[thick] (11.2,2) .. controls (11.6,2.1) .. (12,2) .. controls (12.4,1.9) .. (12.8,2) 
node[pos=0.2, above] {{\footnotesize $\Sigma_1$}};
\end{tikzpicture}
\end{gathered}
\end{flalign}
from a small collar region $V_0$ around $\Sigma_0$, embedded via $i_0$ into $N$,
to a small collar region $V_1$ around $\Sigma_1$, embedded via $i_1$ into $N$.
To avoid confusion, let us emphasize that we have drawn
this bordism vertically, despite the fact that it is a horizontal morphism in the pseudo-category, 
because the time dimension is often drawn vertically (from bottom to top) in Lorentzian geometry.

\item[\underline{$\mathsf{Mor}$:}] Equivalence classes
\begin{subequations}\label{eqn:LBord2cell}
\begin{flalign}
\begin{gathered}
\xymatrix@R=0.25em@C=0.25em{
 (M_0^\prime,\Sigma_0^\prime) \ar[rr]|{\sla}^{(N^\prime,i_0^\prime,i_1^\prime)}&&  (M_1^\prime,\Sigma_1^\prime)\\
 &{\scriptstyle [Z,f]}~\rotatebox[origin=c]{90}{$\Longrightarrow$}~~~&\\
 (M_0,\Sigma_0) \ar[rr]|{\sla}_{(N,i_0,i_1)}&&  (M_1,\Sigma_1)
}
\end{gathered}
\end{flalign}
of zig-zags
\begin{flalign}
\xymatrix{
N &\ar[l]_-{\subseteq} Z \ar[r]^f & N^\prime
}
\end{flalign}
of Cauchy morphisms in $\Loc_m$ satisfying
\begin{flalign}\label{eqn:bordmorphconditions}
J^+_N\big(i_0(\Sigma_0)\big)\cap J^-_N\big(i_1(\Sigma_1)\big)\subset Z~~,\quad
fi_0(\Sigma_0)\,=\,i_0^\prime(\Sigma_0^\prime)~~,\quad
fi_1(\Sigma_1)\,=\,i_1^\prime(\Sigma_1^\prime)\quad.
\end{flalign}
\end{subequations}
The equivalence relation is analogous to the one used for morphisms in $(\LB_m)_0$,
with the condition that $J^+_N\big(i_0(\Sigma_0)\big)\cap J^-_N\big(i_1(\Sigma_1)\big)$ 
must be contained in any smaller choice of $Z$.
\sk

The first condition in \eqref{eqn:bordmorphconditions}
states that the time slab between the two embedded Cauchy surfaces
$i_0(\Sigma_0)$ and $i_1(\Sigma_1)$ in $N$ must be contained in the subset $Z\subseteq N$.
Geometrically, one interprets $[Z,f]$ as the germ of a local Cauchy morphism
\begin{flalign}
\begin{gathered}\begin{tikzpicture}[scale=1]
\draw[dashed,fill=gray!5] (5.2,-1.5) -- (5.2,3.5) -- (6.8,3.5) -- (6.8,-1.5) -- (5.2,-1.5);
\draw[dashed,fill=gray!20] (5.2,-0.5) -- (5.2,2.5) -- (6.8,2.5) -- (6.8,-0.5) -- (5.2,-0.5);
\draw (6,3) node {{\footnotesize $N$}}; 
\draw[thick] (5.2,0) .. controls (5.6,0.1) .. (6,0) .. controls (6.4,-0.1) .. (6.8,0.1) 
node[pos=0.2, above] {{\footnotesize $i_0(\Sigma_0)~~$}};
\draw[thick] (5.2,2) .. controls (5.6,2.1) .. (6,2) .. controls (6.4,1.9) .. (6.8,2) 
node[pos=0.2, above] {{\footnotesize $i_1(\Sigma_1)~~$}};
\draw[->,thick] (7.8,1) -- (7.2,1) node [midway, above] {{\footnotesize $\subseteq$}};
\draw[dashed,fill=gray!20] (8.2,-0.5) -- (8.2,2.5) -- (9.8,2.5) -- (9.8,-0.5) -- (8.2,-0.5);
\draw (9,1) node {{\footnotesize $Z$}}; 
\draw[thick] (8.2,0) .. controls (8.6,0.1) .. (9,0) .. controls (9.4,-0.1) .. (9.8,0.1) 
node[pos=0.2, above] {{\footnotesize $i_0(\Sigma_0)~~$}};
\draw[thick] (8.2,2) .. controls (8.6,2.1) .. (9,2) .. controls (9.4,1.9) .. (9.8,2) 
node[pos=0.2, above] {{\footnotesize $i_1(\Sigma_1)~~$}};
\draw[->,thick] (10.2,1.1) -- (10.8,1.4) node [midway, above] {{\footnotesize $f$}};
\draw[dashed,fill=gray!5] (11.2,-1.5) -- (11.2,4) -- (12.8,4) -- (12.8,-1.5) -- (11.2,-1.5);
\draw[dashed,fill=gray!20] (11.2,0.5) -- (11.2,3.5) -- (12.8,3.5) -- (12.8,0.5) -- (11.2,0.5);
\draw (12,-0.75) node {{\footnotesize $N^\prime$}}; 
\draw[thick] (11.2,1) .. controls (11.6,1.1) .. (12,1) .. controls (12.4,0.9) .. (12.8,1.1) 
node[pos=0.2, above] {{\footnotesize $i_0^\prime(\Sigma_0^\prime)~~$}};
\draw[thick] (11.2,3) .. controls (11.6,3.1) .. (12,3) .. controls (12.4,2.9) .. (12.8,3) 
node[pos=0.2, above] {{\footnotesize $i_1^\prime(\Sigma_1^\prime)~~$}};
\end{tikzpicture}
\end{gathered}
\end{flalign}
that identifies locally isomorphic collar regions around this time slab.
\end{itemize}

Identities and compositions in the groupoid  $(\LB_m)_1$ are defined 
analogously to the groupoid $(\LB_m)_0$ above. 

\paragraph{Source and target:} The source functor is defined by
\begin{subequations}\label{eqn:sourcefunctor}
\begin{flalign}
\nn s\,:\,  (\LB_m)_1~&\longrightarrow~ (\LB_m)_0\quad,\\
\nn \Big((N,i_0,i_1) : (M_0,\Sigma_0)\hto (M_1,\Sigma_1)\Big)~&\longmapsto~(M_0,\Sigma_0)\quad,\\
\begin{gathered}
\xymatrix@R=0.25em@C=0.25em{
 (M_0^\prime,\Sigma_0^\prime) \ar[rr]|{\sla}^{(N^\prime,i_0^\prime,i_1^\prime)}&&  (M_1^\prime,\Sigma_1^\prime)\\
 &{\scriptstyle [Z,f]}~\rotatebox[origin=c]{90}{$\Longrightarrow$}~~~&\\
 (M_0,\Sigma_0) \ar[rr]|{\sla}_{(N,i_0,i_1)}&&  (M_1,\Sigma_1)
}
\end{gathered} ~&\longmapsto~
\begin{gathered}
\xymatrix@R=1em@C=0.25em{
 (M_0^\prime,\Sigma_0^\prime) \\
    ~\\
 \ar[uu]^-{\big[i_0^{-1} \big( f^{-1} \big(i_0^\prime(V_0^\prime)\big)\big), i_0^{\prime -1} \,f\, i_0 \big]} (M_0,\Sigma_0)
}
\end{gathered}\quad,
\end{flalign}
with the action on morphisms given by the factorization
\begin{flalign}
\begin{gathered}
\xymatrix@C=2em{
M_0 &\ar[l]_-{\subseteq} V_0 \ar[r]^-{i_0} & N &\ar[l]_-{\subseteq} Z \ar[r]^-{f}& N^\prime &\ar[l]_-{i_0^\prime} 
\ar[dl]_-{\cong}^-{i_0^\prime} V_0^\prime \ar[r]^-{\subseteq}& M_0^\prime\\
    &  \ar[u]^-{\subseteq}i_0^{-1}\big(f^{-1}\big(i_0^\prime(V_0^\prime)\big)\big)   \ar@{-->}[rr]_-{i_0} &   &   
    \ar[u]^-{\subseteq} f^{-1}\big(i_0^\prime(V_0^\prime)\big) \ar@{-->}[r]_-{f} &   
    \ar[u]^-{\subseteq} i_0^\prime(V_0^\prime)  &            &
}
\end{gathered}\quad.
\end{flalign}
\end{subequations}
The target functor is defined similarly by
\begin{subequations}\label{eqn:targetfunctor}
\begin{flalign}
\nn t\,:\,  (\LB_m)_1~&\longrightarrow~ (\LB_m)_0\quad,\\
\nn \Big((N,i_0,i_1) : (M_0,\Sigma_0)\hto (M_1,\Sigma_1)\Big)~&\longmapsto~(M_1,\Sigma_1)\quad,\\
\begin{gathered}
\xymatrix@R=0.25em@C=0.25em{
 (M_0^\prime,\Sigma_0^\prime) \ar[rr]|{\sla}^{(N^\prime,i_0^\prime,i_1^\prime)}&&  (M_1^\prime,\Sigma_1^\prime)\\
 &{\scriptstyle [Z,f]}~\rotatebox[origin=c]{90}{$\Longrightarrow$}~~~&\\
 (M_0,\Sigma_0) \ar[rr]|{\sla}_{(N,i_0,i_1)}&&  (M_1,\Sigma_1)
}
\end{gathered} ~&\longmapsto~
\begin{gathered}
\xymatrix@R=1.25em@C=0.25em{
 (M_1^\prime,\Sigma_1^\prime) \\
    ~\\
 \ar[uu]^-{\big[i_1^{-1} \big( f^{-1} \big(i_1^\prime(V_1^\prime)\big)\big), i_1^{\prime -1} \,f\, i_1 \big]} (M_1,\Sigma_1)
}
\end{gathered}\quad,
\end{flalign}
with the action on morphisms given by the factorization
\begin{flalign}
\begin{gathered}
\xymatrix@C=2em{
M_1 &\ar[l]_-{\subseteq} V_1 \ar[r]^-{i_1} & N &\ar[l]_-{\subseteq} Z \ar[r]^-{f}& N^\prime &\ar[l]_-{i_1^\prime} 
\ar[dl]_-{\cong}^-{i_1^\prime} V_1^\prime \ar[r]^-{\subseteq}& M_1^\prime\\
    &  \ar[u]^-{\subseteq}i_1^{-1}\big(f^{-1}\big(i_1^\prime(V_1^\prime)\big)\big)   \ar@{-->}[rr]_-{i_1} &   &   
    \ar[u]^-{\subseteq} f^{-1}\big(i_1^\prime(V_1^\prime)\big) \ar@{-->}[r]_-{f} &   
    \ar[u]^-{\subseteq} i_1^\prime(V_1^\prime)  &            &
}
\end{gathered}\quad.
\end{flalign}
\end{subequations}

\paragraph{Horizontal composition:} 
Let us consider any two composable horizontal morphisms, i.e.\ bordisms, 
$(N_0,i_{00},i_{01}) : (M_0,\Sigma_0)\hto (M_1,\Sigma_1)$
and $(N_1,i_{10},i_{11}) : (M_1,\Sigma_1)\hto (M_2,\Sigma_2)$.
Defining the subsets
\begin{subequations}
\begin{flalign}
N_0^-\,&:=\, J^-_{N_0}\big(i_{01}(V_{01}\cap V_{10})\big)\,\subseteq \,N_0\quad,\\
N_1^+\,&:=\, J^+_{N_1}\big(i_{10}(V_{01}\cap V_{10})\big)\,\subseteq \,N_1\quad,
\end{flalign}
\end{subequations}
we obtain a commutative diagram
\begin{flalign}\label{eqn:composedbordismdiagram}
\begin{gathered}
\resizebox{0.9\hsize}{!}{
$\xymatrix@C=1.5em{
M_0 & \ar[l]_-{\subseteq}V_{00}\ar[r]^-{i_{00}} & N_0 & \ar[l]_-{i_{01}} V_{01} \ar[r]^-{\subseteq}& 
M_1 & \ar[l]_-{\subseteq} V_{10} \ar[r]^-{i_{10}}& N_1 &\ar[l]_-{i_{11}} V_{11} \ar[r]^-{\subseteq}& M_2\\
& \ar@/_2.0pc/[ddrrr]_-{\,i_{-}\,i_{00}} i_{00}^{-1}(N_0^-) \ar[u]^-{\subseteq} \ar[rrd]_-{i_{00}}& & &  \ar[lu]^-{\subseteq}  \ar[dl]_-{i_{01}} V_{01}\cap V_{10} \ar[ru]_-{\subseteq} \ar[dr]^-{i_{10}}& & & 
 \ar@/^2.0pc/[ddlll]^-{\,i_{+}\,i_{11}} i_{11}^{-1}(N_1^+\big) \ar[u]_-{\subseteq} \ar[lld]^-{i_{11}} & \\
& & & \ar[uul]^-{\subseteq} N_0^-  \ar@{-->}[dr]^-{i_-}&  & \ar[uur]_-{\subseteq}
N_1^+ \ar@{-->}[dl]_-{i_+} & & & \\
& & & & 
N_0^-
\sqcup^{~}_{V_{01}\cap V_{10}}  
N_1^+  & & & & 
}
$}
\end{gathered}
\end{flalign}
of Cauchy morphisms in $\Loc_m$. 
We define the horizontal composite 
\begin{subequations}\label{eqn:composedbordism}
\begin{flalign}
(N_1,i_{10},i_{11})\cmp(N_0,i_{00},i_{01})\,:
\xymatrix@R=0.25em@C=0.35em{
(M_0,\Sigma_0) \ar[rr]|{\sla}^{}~&~&~ (M_2,\Sigma_2)
}
\end{flalign}
of horizontal morphisms by
\begin{flalign}
(N_1,i_{10},i_{11})\cmp(N_0,i_{00},i_{01})\,&:=\,\Big(
N_0^-
\sqcup^{~}_{V_{01}\cap V_{10}} 
N_1^+,
i_-\,i_{00}, i_{+}\,i_{11}\Big)\quad.
\end{flalign}
\end{subequations}
The geometric picture behind this construction is as follows:
We take the intersection $V_{01}\cap V_{10}\subseteq M_1$
of the two small collar regions of $\Sigma_1$ in the intermediate object $M_1$.
The intersection embeds via $i_{01}$ into the first bordism $N_0$
and we take the causal past $N_0^- = J^-_{N_0}\big(i_{01}(V_{01}\cap V_{10})\big)\subseteq N_0$ of its image.
Similarly, the intersection embeds via $i_{10}$ into the second bordism $N_1$
and we take the causal future $N_1^+ = J^+_{N_1}\big(i_{10}(V_{01}\cap V_{10})\big)\subseteq N_1$ of its image.
The composed bordism \eqref{eqn:composedbordism} is then defined in terms of a pushout, i.e.\ by gluing
the oriented and time-oriented globally hyperbolic Lorentzian manifolds
$N_0^-$ and $N_1^+$ along $V_{01}\cap V_{10}$.
\sk

Given two horizontally composable $2$-cells
\begin{flalign}
\begin{gathered}
\xymatrix@R=0.25em@C=0.25em{
 (M_0^\prime,\Sigma_0^\prime) \ar[rr]|{\sla}^{(N_0^\prime,i_{00}^\prime,i_{01}^\prime)}&&  (M_1^\prime,\Sigma_1^\prime)
\ar[rr]|{\sla}^{(N_1^\prime,i_{10}^\prime,i_{11}^\prime)}   &&  (M_2^\prime,\Sigma_2^\prime)\\
 &{\scriptstyle [Z_0,f_0]}~\rotatebox[origin=c]{90}{$\Longrightarrow$}~~~& &{\scriptstyle [Z_1,f_1]}~\rotatebox[origin=c]{90}{$\Longrightarrow$}~~~&\\
\ar[uu]^-{s[Z_0,f_0]} (M_0,\Sigma_0) \ar[rr]|{\sla}_{(N_0,i_{00},i_{01})}&&  \ar[uu] (M_1,\Sigma_1) 
\ar[rr]|{\sla}_{(N_1,i_{10},i_{11})} &&  (M_2,\Sigma_2) \ar[uu]_-{t[Z_1,f_1]}
}
\end{gathered}\qquad,
\end{flalign}
i.e.\ $t[Z_0,f_0] = s[Z_1,f_1]$, there exists, by definition of source \eqref{eqn:sourcefunctor} 
and target \eqref{eqn:targetfunctor}, a causally convex open subset 
$W\subseteq V_{01}\cap V_{10}\subseteq M_1$, containing the Cauchy surface $\Sigma_1\subset W$,
such that $f^\cap := i_{01}^{\prime -1} \,f_0\, i_{01}\big\vert_{W} =
i_{10}^{\prime -1} \,f_1\, i_{10}\big\vert_{W}$. Defining the subsets
\begin{subequations}
\begin{flalign}
Z^- \,&:=\, N_0^{-}\cap f_0^{-1}(N_0^{\prime -})\,\subseteq\,N_0\quad,\\
Z^+\,&:=\,  N_1^+ \cap f_1^{-1}(N_1^{\prime +})\,\subseteq N_1\quad,\\
Z^\cap\,&:=\, i_{01}^{-1}(Z^-)\cap i_{10}^{-1}(Z^+) \cap W\,\subseteq M_1\quad,
\end{flalign}
\end{subequations}
we obtain a commutative diagram
\begin{flalign}
\begin{gathered}
\xymatrix{
N_0^{\prime -} &
\ar[l]_-{i_{01}^\prime} V_{01}^\prime \cap V_{10}^\prime  \ar[r]^-{i_{10}^\prime}&
N_1^{\prime +}\\
\ar[d]_-{\subseteq} Z^- \ar[u]^-{f_0}& 
\ar[d]_-{\subseteq} \ar[l]_-{i_{01}}   Z^{\cap} \ar[r]^-{i_{10}} \ar[u]^-{f^\cap}& 
\ar[d]^-{\subseteq} Z^+ \ar[u]_-{f_1}\\
N_0^- & \ar[l]^-{i_{01}} V_{01}\cap V_{10}  \ar[r]_-{i_{10}}& 
N_1^+
}
\end{gathered}
\end{flalign}
of Cauchy morphisms in $\Loc_m$. This induces morphisms
\begin{flalign}
\xymatrix@C=3.5em{
N_0^{-}\sqcup_{V_{01}\cap V_{10}}^{~} N_1^+ &\ar[l]_-{\subseteq}
Z^-\sqcup_{Z^\cap}^{~} Z^+ \ar[r]^-{f_0\sqcup_{f^\cap}^{} f_1} &
N_0^{\prime -} \sqcup_{V_{01}^\prime \cap V_{10}^\prime}^{~} N_1^{\prime +}
}
\end{flalign}
between the pushouts of the rows, which we use to define the horizontal composite 
\begin{subequations}\label{eqn:composed2cell}
\begin{flalign}
[Z_1,f_1]\cmp[Z_0,f_0] \,:\, (N_1,i_{10},i_{11})\cmp(N_0,i_{00},i_{01})~\Longrightarrow~
(N_1^\prime,i_{10}^\prime,i_{11}^\prime)\cmp(N_0^\prime,i_{00}^\prime,i_{01}^\prime)
\end{flalign}
of $2$-cells by
\begin{flalign}
[Z_1,f_1]\cmp[Z_0,f_0]\,:=\, \big[Z^-\sqcup_{Z^\cap}^{~} Z^+, f_0\sqcup_{f^\cap}^{} f_1\big]\quad.
\end{flalign}
\end{subequations}

\paragraph{Horizontal unit:} We define the horizontal unit functor by
\begin{flalign}
\nn u\,:\, (\LB_m)_0~&\longrightarrow~(\LB_m)_1\quad,\\
\nn (M,\Sigma)~&\longmapsto~ \Big((M,\id,\id) : (M,\Sigma)\to(M,\Sigma)\Big)\quad,\\
\begin{gathered}
\xymatrix@R=1.25em@C=0.25em{
 (M^\prime,\Sigma^\prime) \\
    ~\\
 \ar[uu]^-{[W,g]} (M,\Sigma)
}
\end{gathered}~&\longmapsto~
\begin{gathered}
\xymatrix@R=0.25em@C=0.25em{
 (M^\prime,\Sigma^\prime) \ar[rr]|{\sla}^{(M^\prime,\id,\id)}&&  (M^\prime,\Sigma^\prime)\\
 &{\scriptstyle [W,g]}~\rotatebox[origin=c]{90}{$\Longrightarrow$}~~~&\\
 (M,\Sigma) \ar[rr]|{\sla}_{(M,\id,\id)}&&  (M,\Sigma)
}
\end{gathered}
\quad.\label{eqn:identitybordism}
\end{flalign}

\paragraph{Coherence isomorphisms:} The horizontal
composition of bordisms is \textit{not} strictly associative and unital 
due to potential mismatches of the gluing region $V_{01}\cap V_{10}\subseteq M_1$ 
and the collar regions $i_{00}^{-1}(N_0^-)\subseteq M_0$ and $i_{11}^{-1}(N_1^+)\subseteq M_2$
in the defining diagram \eqref{eqn:composedbordismdiagram}, as well as due to the fact that
pushouts $N_0^-\sqcup_{V_{01}\cap V_{10}}^{~} N_1^+$ are only defined uniquely up to canonical isomorphism.
We will now show that these potential mismatches lead to canonical globular isomorphisms
between the corresponding bordisms. This allows us to define the coherence isomorphisms
$(\mathsf{a},\mathsf{l},\mathsf{r})$ of the pseudo-category $\LB_m$.
\sk

Let us start with the following observation: For
any bordism $(N,i_0,i_1) : (M_0,\Sigma_0)\hto (M_1,\Sigma_1)$ in $\LB_m$,
one can find causally convex open subsets $\widetilde{V}_0\subseteq V_0$ and 
$\widetilde{V}_1\subseteq V_1$ that contain the marked Cauchy surfaces, i.e.\ 
$\Sigma_0\subset \widetilde{V}_0$ and $\Sigma_{1}\subset \widetilde{V}_1$,
such that $i_0(\widetilde{V}_0)\subseteq J_N^-\big(i_1(\widetilde{V}_1)\big)$
and $i_1(\widetilde{V}_1)\subseteq J_N^+\big(i_{0}(\widetilde{V}_0)\big)$.
(For example, take $\widetilde{V}_0 := i_0^{-1}\big(J^-_N\big(i_1(V_1)\big) \big)\subseteq M_0$
and $\widetilde{V}_1 := i_1^{-1}\big(J^+_N\big(i_0(V_0)\big) \big)\subseteq M_1$.)
For any causally convex open subset $\widetilde{N}\subseteq N$ satisfying
\begin{subequations}\label{eqn:smallercollars}
\begin{flalign}
J_N^+\big(i_0(\widetilde{V}_0)\big) \cap J_N^-\big(i_1(\widetilde{V}_1)\big) \,\subseteq\, 
\widetilde{N} \,\subseteq\,N\quad,
\end{flalign}
the associated commutative diagram
\begin{flalign}
\begin{gathered}
\xymatrix{
M_0 & \ar[l]_-{\subseteq} V_0 \ar[r]^-{i_0}& N &\ar[l]_-{i_1} V_1 \ar[r]^-{\subseteq}& M_1\\
\ar[u]^-{\id} M_0 & \ar[l]^-{\subseteq}\widetilde{V}_0 \ar[u]^-{\subseteq}\ar[r]_-{i_0} & \ar[u]^-{\subseteq} \widetilde{N} & \ar[u]^-{\subseteq} \ar[l]^-{i_1}\widetilde{V}_1 \ar[r]_-{\subseteq}&\ar[u]_-{\id} M_1
}
\end{gathered}
\end{flalign}
defines a globular $2$-cell isomorphism 
\begin{flalign}
(\widetilde{N},i_0,i_1) ~\cong~ (N,i_0,i_1)
\end{flalign} 
\end{subequations}
in $\LB_m$.
This means that, up to a canonical isomorphism given by a globular $2$-cell, we can always make the collar
regions $V_0\subseteq M_0$ and $V_1\subseteq M_1$ around the Cauchy surfaces arbitrarily small
and remove the unnecessary parts of the collar region around the bordism.
\sk

Let us consider now the horizontal composition 
$(N_1,i_{10},i_{11})\cmp(N_0,i_{00},i_{01})$ of two bordisms
from \eqref{eqn:composedbordism}. Given  
any causally convex open subset $W\subseteq V_{01}\cap V_{10}$ that contains the marked
Cauchy surface, i.e.\ $\Sigma_1\subset W$, the associated commutative diagram
\begin{subequations}
\begin{flalign}
\begin{gathered}
\xymatrix{
N_0^- &\ar[l]_-{i_{01}} V_{01}\cap V_{10} \ar[r]^-{i_{10}}& N_1^+\\
\ar[u]^-{\subseteq} J^-_{N_0}\big(i_{01}(W)\big) &\ar[l]^-{i_{01}} W \ar[r]_-{i_{10}}\ar[u]^-{\subseteq} &\ar[u]_-{\subseteq} J^+_{N_1}\big(i_{10}(W)\big)
}
\end{gathered}
\end{flalign}
induces a canonical isomorphism between the pushouts
\begin{flalign}
J^-_{N_0}\big(i_{01}(W)\big)\sqcup_{W}^{~} J^+_{N_1}\big(i_{10}(W)\big) ~\stackrel{\cong}{\longrightarrow}~
N_0^-\sqcup_{ V_{01}\cap V_{10}}^{~} N_1^+ \quad.
\end{flalign}
\end{subequations}
Together with \eqref{eqn:smallercollars}, this implies that the
composite bordism $(N_1,i_{10},i_{11})\cmp(N_0,i_{00},i_{01})$ from \eqref{eqn:composedbordism}
is canonically isomorphic, via a globular $2$-cell 
\begin{subequations}\label{eqn:bordismarbitrary}
\begin{flalign}
(N_1,i_{10},i_{11})\cmp(N_0,i_{00},i_{01})~\cong~
\Big(J^-_{N_0}\big(i_{01}(W)\big)\sqcup_{W}^{~} J^+_{N_1}\big(i_{10}(W)\big),\,i_-\,i_{00},\,i_+\,i_{11}\Big)\quad,
\end{flalign}
to the bordism that is defined by the chain of zig-zags
\begin{flalign}
\xymatrix@C=3em{
M_0 & \ar[l]_-{\subseteq} \widetilde{V}_0  \ar[r]^-{i_- i_{00}}& 
J^-_{N_0}\big(i_{01}(W)\big)\sqcup_{W}^{~} J^+_{N_1}\big(i_{10}(W)\big) & 
\ar[l]_-{i_+ i_{11}} \widetilde{V}_2\ar[r]^-{\subseteq} & M_2
}
\end{flalign}
\end{subequations}
of Cauchy morphisms in $\Loc_m$ for arbitrarily small causally convex open subsets
$W\subseteq V_{01}\cap V_{10}$,
$\widetilde{V}_0\subseteq i_{00}^{-1}\big(J^-_{N_0}\big(i_{01}(W)\big)\big)$ and
$\widetilde{V}_2\subseteq i_{11}^{-1}\big(J^+_{N_1}\big(i_{10}(W)\big)$
that contain the marked Cauchy surfaces, i.e.\ 
$\Sigma_1\subset W$, $\Sigma_0\subset \widetilde{V}_0$ and $\Sigma_{2}\subset \widetilde{V}_2$.
\sk

With these preparations we can now define the coherence isomorphisms
$(\mathsf{a},\mathsf{l},\mathsf{r})$.
Using \eqref{eqn:smallercollars} and \eqref{eqn:bordismarbitrary}, 
the components of the associator $\mathsf{a}$ are given by the
canonical globular $2$-cells
\begin{flalign}\label{eqn:associatorLBord}
\nn &\big((N_2,i_{20},i_{21})\cmp(N_1,i_{10},i_{11})\big) \cmp (N_0,i_{00},i_{01})\\[4pt]
\nn &\quad ~\cong\, \Big(J^-_{N_0}\big(i_{01}(W_1)\big) \sqcup_{W_1} 
\Big(J^+_{N_1}\big(i_{10}(W_1)\big)\cap J^-_{N_1}\big(i_{11}(W_2)\big)\Big) \sqcup_{W_2} 
J^+_{N_2}\big(i_{20}(W_2)\big),\,i_-\,i_{00} ,\,i_+\,i_{21}\Big)\\[4pt]
&\quad~\cong\, (N_2,i_{20},i_{21})\cmp \big( (N_1,i_{10},i_{11}) \cmp (N_0,i_{00},i_{01})\big)\quad,
\end{flalign}
where the causally convex open neighborhoods $W_1\subseteq V_{01}\cap V_{10}$ of $\Sigma_1\subset M_1$
and $W_2\subseteq V_{11}\cap V_{20}$ of $\Sigma_2\subset M_2$ are chosen sufficiently small
such that $i_{11}(W_2)\subseteq J^+_{N_1}\big(i_{10}(W_1)\big)$
and $i_{10}(W_1)\subseteq J^-_{N_1}\big(i_{11}(W_2)\big)$.
\sk

We now consider the left unitor $\mathsf{l}$.
Let $(N, i_0, i_1) : (M_0, \Sigma_0) \hto (M_1, \Sigma_1)$ be a horizontal morphism in $\LB_m$.
By~\eqref{eqn:composedbordismdiagram}, the composition $(M_1, \id, \id) \cmp (N, i_0, i_1)$ with the 
identity can be depicted as the horizontal chain of zig-zags in the diagram
\begin{subequations}
\begin{flalign}
\begin{gathered}
\xymatrix@C=1em{
	& & J_N^-\big(i_1(V_1)\big) \ar[dr]^-{i_-}
	& V_1 \ar[l]_-{i_1} \ar[r]^-{\subseteq}
	& J_{M_1}^+(V_1)  \ar[dl]_-{i_+}
	& & 
	\\
	M_0
	& V^\prime_0 \ar[l]_-{\subseteq} \ar[ur]^-{i_0} \ar[rr]^-{i_- i_0}
	&
	& J_N^-\big(i_1(V_1)\big) \sqcup_{V_1} J_{M_1}^+(V_1)
	&
	& J_{M_1}^+(V_1) \ar[r]^-{\subseteq} \ar[ul]_-{=} \ar[ll]_-{i_+}
	& M_1
} 
\end{gathered}\quad,
\end{flalign}
where we have abbreviated
\begin{flalign}
V^\prime_0 \,:=\, i_0^{-1} \big( J_N^-\big(i_1(V_1)\big) \big)\,\subseteq\,V_0\quad.
\end{flalign}
\end{subequations}
We define a $2$-cell $(M_1, \id, \id) \cmp (N, i_0, i_1) \Rightarrow (N,i_0,i_1)$
in $\LB_m$ by the vertical data in the diagram
\begin{flalign}
\label{eqn:unitorLBordLeft}
\begin{gathered}
\xymatrix@C=1em{
	& & J_N^-(i_1(V_1)) \ar[dr]^-{i_-}
	& V_1 \ar[l]_-{i_1} \ar[r]^-{\subseteq}
	& J_{M_1}^+(V_1)  \ar[dl]_-{i_+}
	& & 
	\\
	M_0
	& V^\prime_0 \ar[l]_-{\subseteq} \ar[ur]^-{i_0} \ar[rr]^-{i_- i_0}
	&
	& J_N^-\big(i_1(V_1)\big) \sqcup_{V_1} J_{M_1}^+(V_1)
	&
	& J_{M_1}^+(V_1) \ar[r]^-{\subseteq} \ar[ul]_-{=} \ar[ll]_-{i_+}
	& M_1
	\\
	& V^\prime_0 \ar[u]^-{=} \ar[d]_-{\subseteq} \ar[rr]^-{i_- i_0}
	& & i_-\big(J_N^-\big(i_1(V_1)\big)\big) \ar[u]^-{\subseteq} \ar[d]_-{j_-}
	& & V_1 \ar[u]_-{\subseteq} \ar[d]^-{=} \ar[ll]_-{i_- i_1}
	& 
	\\
	M_0
	& V_0 \ar[l]^-{\subseteq} \ar[rr]_-{i_0}
	& & N
	& & V_1 \ar[r]_-{\subseteq} \ar[ll]^-{i_1}
	& M_1
} 
\end{gathered}\quad,
\end{flalign}
where $j_-$ is the inverse of the map $i_-$ after corestricting 
the codomain of $i_-$ to its image (recall that $i_-$ is an isomorphism onto its image).
Note that this $2$-cell is globular.
The components of the left unitor $\mathsf{l}$  are given by the globular
$2$-cells
\begin{flalign}
 (M_1,\id,\id)\cmp(N,i_0,i_1)\,\cong\,(N,i_0,i_1)
\end{flalign}
which are defined by the diagram~\eqref{eqn:unitorLBordLeft}.
The components of the right unitor $\mathsf{r}$
are constructed similarly and they are given by globular $2$-cells
\begin{flalign}\label{eqn:unitorLBordRight}
(N,i_0,i_1)\cmp (M_0,\id,\id)\,=\, 
\Big(J^-_{M_0}(V_0)\sqcup_{V_0} J^+_N\big(i_0(V_0)\big),\, i_-,\, i_+\,i_1\Big)
\,\cong\,(N,i_0,i_1)\quad.
\end{flalign}

\begin{propo}\label{prop:fibrancy}
The globally hyperbolic Lorentzian bordism pseudo-category $\LB_m$
defined above is fibrant in the sense of Definition~\ref{def:fibrant}, i.e.\
$\LB_m\in\PsCat^{\mathrm{fib}}$. A companion
for the vertical morphism $[W,g] : (M,\Sigma)\to (M^\prime,\Sigma^\prime)$ is given by
the horizontal morphism $(M^\prime,g,\id) : (M,\Sigma)\hto (M^\prime,\Sigma^\prime)$
defined by
\begin{flalign}
\xymatrix{
M &\ar[l]_-{\subseteq} W \ar[r]^-{g}& M^\prime & \ar[l]_-{\id} M^\prime \ar[r]^-{=}& M^\prime 
}\quad,
\end{flalign}
together with the $2$-cells
\begin{flalign}
\begin{gathered}
\xymatrix@R=0.25em@C=0.25em{
 (M^\prime,\Sigma^\prime) \ar[rr]|{\sla}^{(M^\prime,\id,\id)} &&  (M^\prime,\Sigma^\prime) \\
 &{\scriptstyle [M^\prime,\id]}~\rotatebox[origin=c]{90}{$\Longrightarrow$}~~&\\
\ar[uu]^-{[W,g]} (M,\Sigma) \ar[rr]|{\sla}_{(M^\prime,g,\id)}&&  \ar@{=}[uu] (M^\prime,\Sigma^\prime) 
}
\end{gathered}
\qquad\text{and}\qquad
\begin{gathered}
\xymatrix@R=0.25em@C=0.25em{
 (M,\Sigma) \ar[rr]|{\sla}^{(M^\prime,g,\id)} &&  (M^\prime,\Sigma^\prime) \\
 &{\scriptstyle [W,g]}~\rotatebox[origin=c]{90}{$\Longrightarrow$}~~&\\
\ar@{=}[uu] (M,\Sigma) \ar[rr]|{\sla}_{(M,\id,\id)}&&  \ar[uu]_-{[W,g]}  (M,\Sigma) 
}
\end{gathered}\qquad.
\end{flalign}
\end{propo}
\begin{proof}
Using the definition of the source and target functors in \eqref{eqn:sourcefunctor} 
and \eqref{eqn:targetfunctor}, one checks that the $2$-cells have the stated
source and target, as required for a companion.
It thus remains to prove that the two composition identities \eqref{eqn:companionidentities} hold true.
The vertical composition identity \eqref{eqn:companionidentities1} follows directly from computing the 
vertical composition
$[M^\prime,\id]\,[W,g] = [W,g] : (M,\id,\id)\Rightarrow (M^\prime,\id,\id)$
and recalling the definition of the unit $2$-cell $u([W,g]) = [W,g]$ from \eqref{eqn:identitybordism}.
\sk

Let us consider now the horizontal composition identity \eqref{eqn:companionidentities2}.
Using the explicit formulas for horizontal compositions of bordisms \eqref{eqn:composedbordism} 
and for $2$-cells \eqref{eqn:composed2cell}, we compute the horizontal composition of the two
middle squares in \eqref{eqn:companionidentities2} and find the $2$-cell
\begin{flalign}
\nn &[M^\prime,\id]\cmp [W,g]\,=\,\Big[W\sqcup_{W}J^+_{M^\prime}\big(g(W)\big),\,g\sqcup_g^{} \id \Big]\,:\, \\[4pt]
&\qquad \Big(J_M^{-}(W)\sqcup_W J^+_{M^\prime}\big(g(W)\big),\,i_-,\, i_+\Big)~\Longrightarrow~
\Big(M^\prime\sqcup_{M^\prime} M^\prime,\,i_-\,g ,\, i_+\Big)\quad.
\end{flalign}
The horizontal composition identity \eqref{eqn:companionidentities2}
then follows by composing this vertically with the relevant components of
the left and right unitors from \eqref{eqn:unitorLBordLeft} and \eqref{eqn:unitorLBordRight}.
Using the explicit description of the unitors, one checks that the resulting $2$-cell is 
the identity $2$-cell.
\end{proof}

For later use, we make the following observation.
\begin{lem}\label{lem:companioninverse}
Let $[W,g] : (M,\Sigma)\to (M^\prime,\Sigma^\prime)$ be any vertical morphism
in $\LB_m$. Then the horizontal morphism $(M^\prime,g,\id) : (M,\Sigma)\hto(M^\prime,\Sigma^\prime)$
given by our choice of companion in Proposition \ref{prop:fibrancy} is horizontally weakly invertible
by the horizontal morphism $(M^\prime,\id,g): (M^\prime,\Sigma^\prime)\hto(M,\Sigma)$ defined by
\begin{flalign}
\xymatrix{
M^\prime &\ar[l]_-{=} M^\prime \ar[r]^-{\id}& M^\prime & \ar[l]_-{g} W \ar[r]^-{\subseteq}& M
}\quad,
\end{flalign}
i.e.\ there exist globular $2$-cells such that
\begin{flalign}
(M^\prime,\id,g) \cmp (M^\prime,g,\id)\,\cong\, (M,\id,\id)\quad\text{and}\quad
(M^\prime,g,\id)\cmp (M^\prime,\id,g) \,\cong\, (M^\prime,\id,\id)\quad.
\end{flalign}
\end{lem}
\begin{proof}
Using \eqref{eqn:composedbordism}, one computes
\begin{flalign}
(M^\prime,\id,g) \cmp (M^\prime,g,\id)\,=\, \Big(M^\prime\sqcup_{M^\prime}M^\prime, i_-\,g, i_+\,g\Big)
\,\cong\,(M^\prime,g,g)\,\cong\, (M,\id,\id)\quad.
\end{flalign}
In the second step we have identified the identity pushout with $M^\prime$
and in the last step we have used the globular $2$-cell
$[W,g]: (M,\id,\id)\Rightarrow (M^\prime,g,g)$.
\sk

Concerning the other composition, we use \eqref{eqn:composedbordism} again and compute
\begin{flalign}
(M^\prime,g,\id)\cmp (M^\prime,\id,g) \,=\,
\Big(J^-_{M^\prime}\big(g(W)\big) \sqcup_{W} J^+_{M^\prime}\big(g(W)\big),\,i_-,\, i_+\Big)\,\cong\,(M^\prime,\id,\id)\quad.
\end{flalign}
In the last step we have identified the pushout with $M^\prime$ and used 
the canonical globular $2$-cells from \eqref{eqn:smallercollars} to resize the collar regions.
\end{proof}


\section{\label{sec:comparison}Comparison theorems between AQFT and FFT}
In this section we prove comparison theorems between algebraic quantum field theories (AQFTs)
and functorial field theories (FFTs) on the globally hyperbolic Lorentzian bordism pseudo-category
from Section~\ref{sec:LBord}. As target category for our quantum field theories, 
we take one of the following categories of $\ast$-algebras in an involutive symmetric monoidal category,
see e.g.\ \cite{Jacobs,BSWinvolutive}.
\begin{defi}\label{def:AlgCat}
We fix an involutive symmetric monoidal category $\TT$ and a 
(not necessarily full) subcategory $\Alg \subseteq {}^{\ast}\Alg_{\mathsf{uAs}}(\TT)$
of the category of associative and unital $\ast$-algebras in $\TT$.
\end{defi}

\begin{ex}
This definition covers the standard choices of target categories in AQFT.
The category ${}^\ast\Alg_{\bbC}=
{}^{\ast}\Alg_{\mathsf{uAs}}(\Vec_{\bbC})$ of 
associative and unital $\ast$-algebras over $\bbC$  
is obtained by choosing the involutive symmetric 
monoidal category $\TT=\Vec_{\bbC}$ of complex vector spaces
with involution functor $\Vec_\bbC\to \Vec_{\bbC}\,,~V\mapsto \overline{V}$
given by complex conjugation of vector spaces and linear maps.
Choosing instead the involutive symmetric monoidal category 
$\TT=\mathbf{Ban}_{\bbC}$ of Banach spaces over $\bbC$, again with involution given by
complex conjugation, we obtain the category 
$C^\ast\Alg_{\bbC}\subseteq {}^{\ast}\Alg_{\mathsf{uAs}}(\mathbf{Ban}_{\bbC})$
of $C^\ast$-algebras as a full subcategory of the category of Banach $\ast$-algebras.
Furthermore, von Neumann algebras and normal unital $\ast$-homomorphisms form a
(non-full) subcategory $W^\ast\mathbf{Alg}_\bbC \subseteq C^\ast\Alg_{\bbC}
\subseteq {}^{\ast}\Alg_{\mathsf{uAs}}(\mathbf{Ban}_{\bbC})$.
\end{ex}

\begin{defi}\label{def:AQFT}
An \textit{$m$-dimensional algebraic quantum field theory} (AQFT) is a functor $\AAA : \Loc_m\to \Alg$
from the category $\Loc_m$ of globally hyperbolic Lorentzian manifolds (see Definition~\ref{def:Loc})
to the category $\Alg$ in Definition~\ref{def:AlgCat} which satisfies
the following properties:
\begin{itemize}
\item[(1)] \textit{Einstein causality:} For all causally disjoint pairs $f_1 : M_1\to M \leftarrow M_2 : f_2$ 
of $\Loc_m$-morphisms (see Definition~\ref{def:Locspecialmorphisms}), the diagram
\begin{flalign}
\begin{gathered}
\xymatrix@C=5em{
\AAA(M_1) \otimes \AAA(M_2)\ar[d]_-{\AAA(f_1)\otimes \AAA(f_2)} \ar[r]^-{\AAA(f_1)\otimes \AAA(f_2)}&\AAA(M)\otimes \AAA(M)\ar[d]^-{\mu_{M}}\\
\AAA(M)\otimes \AAA(M) \ar[r]_-{\mu_{M}^{\op}}&\AAA(M)
}
\end{gathered}
\end{flalign}
commutes, where $\mu_{M}^{(\op)}$ denotes the (opposite) multiplication on the algebra $\AAA(M)$.

\item[(2)] \textit{Time-slice axiom:} For all Cauchy morphisms $f:M\to M^\prime$ in $\Loc_m$ 
(see Definition~\ref{def:Locspecialmorphisms}), the $\Alg$-morphism
\begin{flalign}
\AAA(f)\,:\,\AAA(M)~\stackrel{\cong}{\longrightarrow}~\AAA(M^\prime)
\end{flalign}
is an isomorphism.
\end{itemize}
We denote by 
\begin{flalign}
\AQFT_m\,\subseteq \, \Fun(\Loc_m,\Alg)
\end{flalign} 
the subgroupoid of all $m$-dimensional AQFTs and natural isomorphisms.
\end{defi}

\begin{defi}\label{def:FFT}
An \textit{$m$-dimensional globally hyperbolic Lorentzian functorial field theory} (FFT)
is a pseudo-functor $\FFF : \LB_m\to\iota(\Alg)$ from the globally hyperbolic Lorentzian
bordism pseudo-category $\LB_m$ (see Section~\ref{sec:LBord}) to the pseudo-category
$\iota(\Alg)$ that is obtained by applying Construction \ref{constr:iota} to the category
$\Alg$ in Definition~\ref{def:AlgCat}. We denote by
\begin{flalign}
\FFT_m\,:=\, \PsFun\big(\LB_m,\iota(\Alg)\big)
\end{flalign}
the groupoid of all $m$-dimensional FFTs and (vertical) transformations (see Definition~\ref{def:Transformation}).
\end{defi}

\begin{rem}
The choice of target pseudo-category $\iota(\Alg)$ for FFTs in Definition~\ref{def:FFT}
may seem somewhat unusual from the perspective of functorial or topological field theories.
Here one often interprets the values of an $m$-dimensional
field theory on $(m{-}1)$-dimensional manifolds as \textit{state spaces}, 
with the $m$-dimensional bordisms acting on these state spaces, see e.g.\
\cite{Witten,Atiyah,Segal}.
In contrast, our variant of FFTs describes the assignment of \textit{algebras of 
quantum observables} to Cauchy surfaces with collar regions $(M,\Sigma)\in(\LB_m)_0$,
together with an action of globally hyperbolic Lorentzian bordisms on these algebras.
We explain in Section~\ref{sec:example} how simple examples
of quantum field theories, in particular the free scalar quantum field,
admit a description in terms of an FFT valued in $\iota(\Alg)$.
\sk

Our approach is, however, compatible with the observation
that, in a topological quantum field theory (TQFT) with state-observable correspondence, 
the value of the TQFT on the $(m{-}1)$-sphere 
carries the structure of an algebra, which goes back to the original article~\cite[Section~8]{Segal:CFT_original}.
(See also~\cite{Scheimbauer:Thesis} for TQFTs valued in higher algebras.)
Indeed, the pair-of-pants bordism provides a canonical multiplication, 
while the cap bordism from the empty set to the $(m{-}1)$-sphere specifies a unit element.
For TQFTs this is currently being developed further to 
operators supported on $(m{-}1)$-manifolds of different shapes, see e.g.\ \cite{FreedMooreTeleman}.
\end{rem}

\begin{rem}
Functorial field theories are usually defined as \textit{symmetric monoidal} 
functors between symmetric monoidal categories where, 
on the bordism side, the symmetric monoidal structure consists of 
disjoint union of manifolds. It is important to note that, in the present setting, 
there are no topology-changing bordisms since we require the underlying 
Lorentzian manifolds to be globally hyperbolic, i.e.~admit Cauchy surfaces (see also Remark~\ref{rem:GlobHyp and topology}).
As a consequence, each bordism in $\LB_m$ is topologically a cylinder, 
and so at the level of FFTs disjoint spacetimes do not interact.
This means that there is no need for us to deal with 
the symmetric monoidal structure on bordisms that is
given by disjoint union. 
\end{rem}

Using Theorem~\ref{theo:2adjunction}, we obtain the following
simplified description of FFTs which 
considerably streamlines the proofs of our comparison theorems below.
\begin{cor}\label{cor:FFTsimple}
There exists an equivalence (in fact, an isomorphism)
\begin{flalign}
\FFT_m\,\cong\,\Fun\big(\tau(\LB_m),\Alg\big)
\end{flalign}
between the groupoid of $m$-dimensional FFTs from Definition~\ref{def:FFT}
and the groupoid of ordinary functors from the category $\tau(\LB_m)$
(see Construction \ref{constr:tau}) to $\Alg$ and their natural isomorphisms.
\end{cor}

In what follows we will always suppress this isomorphism 
and simply identify $\FFT_m=\Fun\big(\tau(\LB_m),\Alg\big)$.
\begin{rem}
One should not misinterpret Corollary \ref{cor:FFTsimple} as a statement
that Stolz-Teichner-type geometric bordism pseudo-categories as 
in \cite{StolzTeichner} and Section \ref{sec:LBord}
are not needed to formalize FFTs.
The truncated pseudo-category $\tau(\LB_m)$ still carries the relevant
information about collar regions around Cauchy surfaces
since it has the same objects $(M,\Sigma)\in \tau(\LB_m)$ as the pseudo-category $\LB_m$.
Such information is important for gluing bordisms endowed with geometric data, 
including the bordisms with Lorentzian metrics we consider in this paper. It is 
typically not considered in more naive descriptions of bordism categories 
which work with manifolds with boundaries. The result of Corollary \ref{cor:FFTsimple}
that FFTs admit an equivalent description in terms of ordinary functors 
is a non-trivial consequence of the adjunction from Appendix \ref{app:pseudocats}.
\end{rem} 

Our next goal is to construct a functor
\begin{flalign}\label{eqn:AQFT2FFT}
\FFF_{(-)}^{}\,:\,\AQFT_m~\longrightarrow~\FFT_m
\end{flalign}
that assigns to each AQFT an underlying FFT.
\begin{constr}\label{constr:AQFT2FFT}
We first define the functor \eqref{eqn:AQFT2FFT} on objects.
Given any $\AAA\in \AQFT_m$, we define $\FFF_{\AAA}^{}\in\FFT_m$
in terms of the following functor $\FFF_{\AAA}^{} : \tau(\LB_m)\to\Alg$:
To an object $(M,\Sigma)\in \tau(\LB_m)$, this functor assigns the algebra
\begin{flalign}
\FFF_{\AAA}^{}(M,\Sigma)\,:=\, \AAA(M)\,\in\,\Alg\quad.
\end{flalign}
To a morphism $\big[(N,i_0,i_1) : (M_0,\Sigma_0)\hto (M_1,\Sigma_1)\big]$ in $\tau(\LB_m)$,
i.e.\ an equivalence class under globular $2$-cells of a globally hyperbolic Lorentzian
bordism \eqref{eqn:LBord1cell} represented by the Cauchy morphisms
\begin{flalign}
\xymatrix{
M_0 & \ar[l]_-{\subseteq} V_0 \ar[r]^-{i_0}& N & \ar[l]_-{i_1}V_1 \ar[r]^-{\subseteq}& M_1
}\quad,
\end{flalign}
this functor assigns the $\Alg$-morphism 
\begin{subequations}\label{eqn:FFF_AAAmorphism}
\begin{flalign}
\FFF_{\AAA}^{}\big([N,i_0,i_1]\big) \,:\, 
\FFF_{\AAA}^{}(M_0,\Sigma_0) = \AAA(M_0)~\longrightarrow~ \AAA(M_1)=\FFF_{\AAA}^{}(M_1,\Sigma_1)
\end{flalign}
defined by the commutative diagram
\begin{flalign}
\begin{gathered}
\xymatrix@C=3em{
\AAA(M_0) \ar@{-->}[rr]^-{\FFF_{\AAA}^{}([N,i_0,i_1])}&&  \AAA(M_1)\\
\ar[u]_-{\cong}\AAA(V_0)\ar[r]_-{\AAA(i_0)}^-{\cong} & \AAA(N) &\ar[l]^-{\AAA(i_1)}_-{\cong}\AAA(V_1)\ar[u]^-{\cong}
}
\end{gathered}\quad.
\end{flalign}
\end{subequations}
Note that, as a consequence of the time-slice axiom from Definition~\ref{def:AQFT}, this diagram
consists of $\Alg$-isomorphisms. The $\Alg$-morphism $\FFF_{\AAA}^{}\big([N,i_0,i_1]\big)$
does not depend on the choice of representative of the equivalence class 
$\big[(N,i_0,i_1) : (M_0,\Sigma_0)\hto (M_1,\Sigma_1)\big]$.
Indeed, given any other representative, we obtain from 
the definition of $2$-cells \eqref{eqn:LBord2cell} in $\LB_m$ a commutative diagram
\begin{flalign}
\begin{gathered}
\xymatrix{
M_0 & \ar[l]_-{\subseteq} V^\prime_0 \ar[r]^-{i_0^\prime}& N^\prime & \ar[l]_-{i_1^\prime} V^\prime_1 \ar[r]^-{\subseteq}& M_1\\
 & W_0\ar[d]_-{\subseteq} \ar[u]^-{\subseteq}\ar[r]^-{i_0}& Z\ar[d]_-{\subseteq} \ar[u]^-{f}& 
 W_1\ar[d]^-{\subseteq} \ar[l]_-{i_1} \ar[u]_-{\subseteq}& \\
M_0 &\ar[l]^-{\subseteq} V_0 \ar[r]_-{i_0}& N &\ar[l]^-{i_1} V_1 \ar[r]_-{\subseteq}& M_1
}
\end{gathered}
\end{flalign}
of Cauchy morphisms in $\Loc_m$, where the causally convex opens $W_0$ and $W_1$ with the desired inclusion 
properties exist as a consequence of globularity of the $2$-cell and the definition 
of source \eqref{eqn:sourcefunctor} and target \eqref{eqn:targetfunctor} in $\LB_m$.
Applying the AQFT functor $\AAA : \Loc_m\to \Alg$ to this diagram 
and using the time-slice axiom implies independence of \eqref{eqn:FFF_AAAmorphism} 
of the choice of representative. We observe that the assignment $\FFF_{\AAA}^{} : \tau(\LB_m)\to\Alg$
is a functor since it preserves identities \eqref{eqn:identitybordism} and compositions \eqref{eqn:composedbordism}. 
The latter can be shown by applying the AQFT functor $\AAA : \Loc_m\to \Alg$ to the commutative diagram 
\eqref{eqn:composedbordismdiagram} of Cauchy morphisms in $\Loc_m$ 
that defines the composition of bordisms.
\sk

We now define the functor \eqref{eqn:AQFT2FFT} on morphisms.
Given any morphism $\zeta : \AAA\Rightarrow \BBB$ in $\AQFT_m$, i.e.\ a natural isomorphism
with components $\zeta_M : \AAA(M)\to \BBB(M)$ in $\Alg$, for all $M\in\Loc_m$,
we define the components
\begin{flalign}\label{eqn:FFF_zeta-components}
(\FFF_{\zeta}^{})_{(M,\Sigma)}\,:=\, \zeta_M \,:\,\FFF_{\AAA}^{}(M,\Sigma) \,=\, \AAA(M)~\longrightarrow~\BBB(M)\,=\,\FFF_{\BBB}^{}(M,\Sigma)
\end{flalign}
in $\Alg$, for all $(M,\Sigma)\in \tau(\LB_m)$. These components are natural
with respect to all morphisms $\big[(N,i_0,i_1) : (M_0,\Sigma_0)\hto (M_1,\Sigma_1)\big]$ in $\tau(\LB_m)$
as a consequence of the commutative diagram
\begin{flalign}
\begin{gathered}
\xymatrix@C=3em{
\ar@/^2pc/[rrrr]^-{\FFF_{\BBB}^{}([N,i_0,i_1])}\BBB(M_0) & \ar[l]^-{\cong}\BBB(V_0) \ar[r]_-{\cong}^-{\BBB(i_0)}& \BBB(N) &\ar[l]^-{\cong}_-{\BBB(i_1)} \BBB(V_1) \ar[r]_-{\cong}& \BBB(M_1)\\
\ar@/_2pc/[rrrr]_-{\FFF_{\AAA}^{}([N,i_0,i_1])} \ar[u]^-{\zeta_{M_0}}\AAA(M_0) & \ar[u]^-{\zeta_{V_0}} \ar[l]_-{\cong} \AAA(V_0) \ar[r]^-{\cong}_{\AAA(i_0)}& \ar[u]^-{\zeta_N}\AAA(N) & \ar[u]^-{\zeta_{V_1}} \ar[l]_-{\cong}^-{\AAA(i_1)} \AAA(V_1) \ar[r]^-{\cong}& \ar[u]_-{\zeta_{M_1}}\AAA(M_1)
}
\end{gathered}\qquad,
\end{flalign}
hence they define a morphism $\FFF_{\zeta}^{}: \FFF_{\AAA}^{}\Rightarrow \FFF_{\BBB}^{}$ in the groupoid $\FFT_m$. 
Functoriality of \eqref{eqn:AQFT2FFT} is obvious.
\end{constr}

Summing up, we obtain our first comparison result:
\begin{theo}\label{theo:comparison1}
Construction \ref{constr:AQFT2FFT} defines a functor 
$\FFF_{(-)}^{} : \AQFT_m\to \FFT_m$ that assigns to each
AQFT an underlying FFT. This functor is faithful
and its essential image lies in the full subgroupoid
\begin{flalign}\label{eqn:FFTts}
\FFT_m^{\mathrm{t.s.}}\,\subseteq\, \FFT_m
\end{flalign}
of FFTs that satisfy the following variant of the time-slice axiom:
$\FFF\in \FFT_m^{\mathrm{t.s.}}$ assigns 
to each morphism $\big[(N,i_0,i_1) : (M_0,\Sigma_0)\hto (M_1,\Sigma_1)\big]$ in $\tau(\LB_m)$
an $\Alg$-isomorphism $\FFF\big([N,i_0,i_1]\big): \FFF(M_0,\Sigma_0)\stackrel{\cong}{\longrightarrow}\FFF(M_1,\Sigma_1)$.
\end{theo}
\begin{proof}
The first statement is proven in Construction \ref{constr:AQFT2FFT}. The second
statement is evident from \eqref{eqn:FFF_zeta-components} and the third one follows 
from the fact that $\FFF_{\AAA}^{}\big([N,i_0,i_1]\big)$ is defined in terms of
a composition of $\Alg$-isomorphisms, see \eqref{eqn:FFF_AAAmorphism}.
\end{proof}

The functor $\FFF_{(-)}^{}:\AQFT_m\to \FFT^{\mathrm{t.s.}}_m$ from Theorem \ref{theo:comparison1} 
is clearly forgetful in dimension $m\geq 2$ since the bordism pseudo-category $\LB_m$ 
is constructed only out of the (non-full) subcategory
\begin{flalign}\label{eqn:Cau_m}
\Cau_m\,\subseteq\,\Loc_m
\end{flalign}
consisting of all objects $M\in\Loc_m$ and all Cauchy morphisms
$f: M\to M^\prime$ in $\Loc_m$. (Note that dimension $m=1$ is special since $\Cau_1 = \Loc_1$.)
This means that the underlying FFT $\FFF_{\AAA}\in\FFT_m^{\mathrm{t.s.}}$
of $\AAA\in\AQFT_m$ does not capture any of the \textit{spatially local} structure of AQFTs 
arising from non-Cauchy morphisms, such as Einstein causality from Definition \ref{def:AQFT}.
This is compatible with the fact that geometric FFTs in the sense of \cite{StolzTeichner}
are not extended and it suggests that one should interpret the underlying FFT $\FFF_{\AAA}$ as a way 
to extract and describe only those parts of the AQFT $\AAA$ that are related to time evolution.
In future work, we will try to enhance the algebraic structure on the FFT
side in order to capture also the spatially local structure of AQFTs.
\sk

We would like to mathematically substantiate this interpretation 
of the underlying FFT in terms of time evolution. We do this by 
introducing a concept of \textit{spatially global AQFTs}, which 
does not capture spatially local structures, and proving that the 
groupoid of such objects is equivalent to the groupoid $\FFT_m^{\mathrm{t.s.}}$
of FFTs satisfying the time-slice axiom from \eqref{eqn:FFTts}.
More precisely, we define the groupoid
\begin{flalign}\label{eqn:AQFTsg}
\AQFT_m^{\mathrm{s.g.}}\,\subseteq \,\Fun\big(\Cau_m,\Alg\big)
\end{flalign}
of spatially global AQFTs by replacing in Definition \ref{def:AQFT} 
the category $\Loc_m$ by the subcategory $\Cau_m\subseteq \Loc_m$ of 
Cauchy morphisms. Note that the Einstein causality axiom
becomes void in this case since there are no causally disjoint pairs of
Cauchy morphisms, hence an object $\AAA\in \AQFT_m^{\mathrm{s.g.}}$ is simply 
a functor $\AAA : \Cau_m\to \Alg$ that satisfies the time-slice axiom
from Definition \ref{def:AQFT} (2), and morphisms are natural isomorphisms between
such functors. The inclusion $\Cau_m\subseteq \Loc_m$ of categories
defines a (forgetful) pullback functor $\res : \AQFT_m \to \AQFT_m^{\mathrm{s.g.}}$
through which the functor $\FFF_{(-)}^{}$ from Construction 
\ref{constr:AQFT2FFT} factorizes, i.e.\
\begin{flalign}\label{eqn:sgAQFTvFFT}
\begin{gathered}
\xymatrix{
\ar[dr]_-{\res}\AQFT_m \ar[rr]^-{\FFF_{(-)}^{}}&& \FFT_m^{\mathrm{t.s.}}\\
&\AQFT_m^{\mathrm{s.g.}} \ar@{-->}[ru]_-{\FFF_{(-)}^{}}&
}
\end{gathered}\quad.
\end{flalign}
The following comparison result shows that spatially global AQFTs are equivalent
to globally hyperbolic Lorentzian FFTs satisfying the time-slice axiom.
\begin{theo}\label{theo:comparison2}
The dashed functor in \eqref{eqn:sgAQFTvFFT} defines an equivalence
between the groupoid $\AQFT_m^{\mathrm{s.g.}}$ of spatially global AQFTs \eqref{eqn:AQFTsg}
and the groupoid $\FFT_m^{\mathrm{t.s.}}$ of FFTs which satisfy the time-slice axiom from \eqref{eqn:FFTts}.
\end{theo}
\begin{proof}
We have already observed faithfulness of the functor 
$\FFF_{(-)}^{}:\AQFT_m^{\mathrm{s.g.}}\to \FFT_m^{\mathrm{t.s.}}$
in Theorem \ref{theo:comparison1}, so it remains to prove
fullness and essential surjectivity.
\sk

\noindent \underline{\textit{Fullness:}} Suppose that we are given two spatially 
global AQFTs $\AAA,\BBB\in\AQFT_m^{\mathrm{s.g.}}$
and a natural isomorphism $\zeta: \FFF_{\AAA}^{}\Rightarrow\FFF_{\BBB}^{}$ between the 
underlying FFTs, with components denoted by
\begin{flalign}\label{eqn:zeta for fullness}
\zeta_{(M,\Sigma)} \,:\, \FFF_{\AAA}^{}(M,\Sigma) \,=\, \AAA(M)
~\longrightarrow~ \BBB(M) \, =\, \FFF_{\BBB}^{}(M,\Sigma)\quad ,
\end{flalign}
for all $(M,\Sigma)\in\tau(\LB_m)$. We now show that
these components are independent of the choice of Cauchy surface $\Sigma\subset M$. Given two Cauchy
surfaces $\Sigma\subset M$ and $\Sigma^\prime\subset M$ with $\Sigma^\prime\subset J^+_M(\Sigma)$
in the causal future of $\Sigma$, there exists a bordism $\big[(M,\id,\id):(M,\Sigma)\hto(M,\Sigma^\prime)\big]$
in $\tau(\LB_m)$ represented as in \eqref{eqn:LBord1cell} by
\begin{flalign}\label{eqn:bordismCauchymove}
\xymatrix{
M &\ar[l]_-{=} M \ar[r]^-{\id} & M &\ar[l]_-{\id} M \ar[r]^-{=}& M
}\quad.
\end{flalign}
Recalling \eqref{eqn:FFF_AAAmorphism}, naturality of $\zeta$ with respect to such bordisms
implies that $\zeta_{(M,\Sigma)} = \zeta_{(M,\Sigma^\prime)}$ for all Cauchy surfaces satisfying
$\Sigma^\prime\subseteq J^+_M(\Sigma)$. The independence of $\zeta_{(M,\Sigma)}$ from the choice
of Cauchy surface $\Sigma\subset M$ then follows from the fact that, for every pair of Cauchy surfaces
$\Sigma,\Sigma^\prime\subset M$, there exists another Cauchy surface $\Sigma^{\prime\prime}\subset M$
in the common causal future, i.e.\ $\Sigma^{\prime\prime}\subset J^+_M(\Sigma) \cap J^+_M(\Sigma^\prime)$.
See e.g.\ \cite[Lemma 3.2]{FAvsAQFT} for an explicit construction of the Cauchy surface $\Sigma^{\prime\prime}\subset M$.
\sk

It remains to verify that the components $\zeta_{M}:= \zeta_{(M,\Sigma)} : \AAA(M)\to \BBB(M)$
define a natural isomorphism of functors from $\Cau_m$ to $\Alg$, i.e.\ an $\AQFT^{\mathrm{s.g.}}_m$-morphism.
Given any morphism $f : M\to M^\prime$ in $\Cau_m$ and any choice of Cauchy surface
$\Sigma\subset M$, there exists a bordism 
$\big[(M^\prime,f,\id):(M,\Sigma)\hto (M^\prime,f(\Sigma))\big]$
in $\tau(\LB_m)$ that is represented by
\begin{flalign}\label{eqn:bordismCau}
\xymatrix{
M &\ar[l]_-{=} M \ar[r]^-{f} & M^\prime &\ar[l]_-{\id} M^\prime \ar[r]^-{=}& M^\prime
}\quad.
\end{flalign}
By~\eqref{eqn:FFF_AAAmorphism}, we have that
\begin{flalign}
\FFF_{\AAA}\big[(M^\prime,f,\id):(M,\Sigma)\hto (M^\prime,f(\Sigma))\big] \,=\, \AAA(f)
\end{flalign}
and analogously for $\FFF_{\BBB}$.
Naturality of $\zeta:\FFF_{\AAA}^{}\Rightarrow\FFF_{\BBB}^{}$ with respect to such bordisms
then implies the required naturality property $\zeta_{M^\prime} \,\AAA(f) = \BBB(f)\,\zeta_M$.
Observing that the functor $\FFF_{(-)}$ maps this new natural 
isomorphism back to the original morphism~\eqref{eqn:zeta for fullness} in 
$\FFT_m^{\mathrm{t.s.}}$ completes the proof of fullness.
\sk

\noindent \underline{\textit{Essential surjectivity:}} Given any $\FFF\in \FFT_m^{\mathrm{t.s.}}$,
we have to construct $\AAA\in\AQFT_{m}^{\mathrm{s.g.}}$ and a natural isomorphism
$\FFF\cong \FFF_{\AAA}^{} $ of FFTs. To define $\AAA(M)\in\Alg$ for
$M\in\Cau_m$, we build out of the family $\FFF(M,\Sigma)\in\Alg$ of algebras, 
for all Cauchy surfaces $\Sigma\subset M$, a single algebra which is independent of $\Sigma$.
To that end, we consider the category $\mathbf{\Sigma}_M$ given by the poset
of all Cauchy surfaces $\Sigma\subset M$ in $M$ with partial order defined by
\begin{flalign}
\Sigma\leq \Sigma^\prime
\quad :\Longleftrightarrow
\quad \Sigma^\prime\subset J^+_M(\Sigma) \quad.
\end{flalign}
By the same argument as above (see also \cite[Lemma 3.2]{FAvsAQFT}),
this poset is directed. Using \eqref{eqn:bordismCauchymove}, the FFT $\FFF$
restricts to a functor $\FFF_M : \mathbf{\Sigma}_M\to \Alg$
which, due to the time-slice axiom, sends all morphisms to isomorphisms. This is equivalent
to a functor $\FFF_M : \mathbf{\Sigma}_M[\mathrm{All}^{-1}]\to \Alg$ on the localization
of the category $\mathbf{\Sigma}_M$ at all morphisms. The fact that the poset $\mathbf{\Sigma}_M$ is directed
is equivalent to saying that the category it presents is filtered.
That, in turn, implies that the localized category $\mathbf{\Sigma}_M[\mathrm{All}^{-1}]$ 
has a contractible nerve~\cite[Section~1, Corollary~2]{Quillen:HAKT_I}. Hence, the colimit
\begin{subequations}\label{eqn:AAAaverage}
\begin{flalign}
\AAA(M)\,:=\,\colim\Big(\FFF_M : \mathbf{\Sigma}_M[\mathrm{All}^{-1}]\to \Alg\Big)\,\in\,\Alg
\end{flalign}
exists, independently of whether or not the category $\Alg$ is cocomplete,
and the canonical maps
\begin{flalign}
\iota_{\Sigma}\, :\, \FFF(M,\Sigma)~\stackrel{\cong}{\longrightarrow} ~ \AAA(M)
\end{flalign} 
\end{subequations}
in the cocone are $\Alg$-isomorphisms, for all Cauchy surfaces $\Sigma\subset M$.
\sk

We now define $\AAA:\Cau_m\to \Alg$ on morphisms.
Given any morphism $f:M\to M^\prime$ in $\Cau_m$, we define the $\Alg$-morphism
$\AAA(f) :\AAA(M)\to \AAA(M^\prime)$ by the universal property of colimits
and the diagrams
\begin{flalign}\label{eqn:AAAaveragemap}
\begin{gathered}
\xymatrix{
\AAA(M) \ar@{-->}[rr]^-{\AAA(f)}&&\AAA(M^\prime)\\
\ar[u]^-{\iota_\Sigma}\FFF(M,\Sigma) \ar[rr]_-{\FFF([M^\prime,f,\id])}&&\FFF(M^\prime,f(\Sigma))\ar[u]_-{\iota_{f(\Sigma)}}
}
\end{gathered}\quad,
\end{flalign}
for all Cauchy surfaces $\Sigma\subset M$, where the
bordism $(M^\prime,f,\id) : (M,\Sigma)\hto(M^\prime,f(\Sigma))$ was defined in \eqref{eqn:bordismCau}.
Recalling the composition of bordisms \eqref{eqn:composedbordism}, one checks that this defines
a functor $\AAA : \Cau_m\to \Alg$.
The key step here is that, for all Cauchy surfaces $\Sigma \leq \Sigma^\prime$ in $\mathbf{\Sigma}_M$, 
we have a commutative square
\begin{flalign}\label{eqn:coherence for AAA(f)}
\begin{gathered}
\xymatrix{
	(M,\Sigma) \ar[rr]^-{[M^\prime, f, \id]} \ar[d]_-{[M, \id, \id]}
	&& \big( M^\prime, f(\Sigma) \big) \ar[d]^-{[M^\prime, \id, \id]}
	\\
	(M,\Sigma^\prime) \ar[rr]_-{[M^\prime, f, \id]}
	&& \big( M^\prime, f(\Sigma^\prime) \big)
}
\end{gathered}
\end{flalign}
in $\tau(\LB_m)$, which can be shown by using the horizontal composition of bordisms
from \eqref{eqn:composedbordism} and the globular $2$-cells in \eqref{eqn:smallercollars}.
Since the morphisms $\iota_\Sigma$, $\iota_{f(\Sigma)}$ and, due to the time-slice axiom, also 
$\FFF([M^\prime,f,\id])$ are invertible, it follows that $\AAA(f)$ is even an isomorphism in $\Alg$.
Therefore, we have constructed an object $\AAA\in \AQFT_m^{\mathrm{s.g.}}$.
\sk

It remains to construct a natural isomorphism $\zeta : \FFF\Rightarrow \FFF_{\AAA}^{}$.
To achieve this, we consider the components
\begin{flalign}
\zeta_{(M,\Sigma)} \,:=\, \iota_{\Sigma} \,:\, 
	\FFF(M,\Sigma) ~\longrightarrow ~\AAA(M) \,=\, \FFF_{\AAA}^{}(M,\Sigma) \quad,
\end{flalign}
defined by the canonical maps into the colimit \eqref{eqn:AAAaverage}.
As we observed above, these are $\Alg$-isomorphisms.
We have to verify that they assemble into a natural
transformation $\FFF \Rightarrow \FFF_{\AAA}$, i.e.~that they fit into naturality squares
\begin{flalign}
\begin{gathered}
\xymatrix{
	\AAA(M_0) \ar[rr]^-{\FFF_{\AAA}([N, i_0, i_1])}
	&& \AAA(M_1)
	\\
	\FFF(M_0,\Sigma_0) \ar[u]^-{\zeta_{(M_0, \Sigma_0)}} \ar[rr]_-{\FFF([N, i_0, i_1])}
	&&\FFF(M_1, \Sigma_1) \ar[u]_-{\zeta_{(M_1, \Sigma_1)}}
}
\end{gathered}\quad,
\end{flalign}
for all morphisms $\big[(N,i_0,i_1) : (M_0,\Sigma_0)\hto (M_1,\Sigma_1)\big]$
in $\tau(\LB_m)$.
Given such a morphism, whose full data is spelled out in~\eqref{eqn:LBord1cell}, 
we expand the associated naturality square as
\begin{flalign}
\label{eqn:naturality for zeta: FFF --> FFF_AAA}
\begin{gathered}
\resizebox{0.9\hsize}{!}{
$\xymatrix@R=4em@C=4em{
	\AAA(M_0) \ar@/^2pc/[rrrrrr]^-{\FFF_{\AAA}([N,i_0,i_1])}
	& \ar[l]^-{\cong}\AAA(V_0) \ar[rr]_-{\AAA(i_0)}
	& &\AAA(N)
	&& \ar[ll]^-{\AAA(i_1)} \AAA(V_1) \ar[r]_-{\cong}
	& \AAA(M_1)
	\\
	\ar@/_2pc/[rrrrrr]_-{\FFF([N,i_0,i_1])} \ar[u]^-{\iota_{\Sigma_0}}\FFF(M_0,\Sigma_0)
	& \ar[l]_-{\FFF([M_0,\subseteq,\id])} \ar[r]^-{\FFF([N,i_0,\id])} \ar[u]^-{\iota_{\Sigma_0}}\FFF(V_0,\Sigma_0)
	& \FFF(N,i_0(\Sigma_0)) \ar[ur]^-{\iota_{i_0(\Sigma_0)}} \ar[rr]^-{\FFF([N,\id,\id])}
	&& \ar[ul]_-{\iota_{i_1(\Sigma_1)}} \FFF(N,i_1(\Sigma_1))
	& \ar[l]_-{\FFF([N,i_1,\id])} \ar[r]^-{\FFF([M_1,\subseteq,\id])}\ar[u]_-{\iota_{\Sigma_1}} \FFF(V_1,\Sigma_1)
	& \FFF(M_1,\Sigma_1)\ar[u]_-{\iota_{\Sigma_1}}
}
$}
\end{gathered}
\end{flalign}
At the top we have used the construction in~\eqref{eqn:FFF_AAAmorphism} 
of the morphism $\FFF_{\AAA}([N,i_0,i_1])$.
We further observe that the triangle and the four squares commute as 
a consequence of $\iota_\Sigma$ being the canonical morphisms 
establishing the colimiting cocone~\eqref{eqn:AAAaverage} and the 
commutative square \eqref{eqn:AAAaveragemap}.
\sk

To show that the lower part of the diagram commutes too,
we first observe that the left-facing bordisms in the bottom row 
are companions in the sense of Proposition \ref{prop:fibrancy}.
By Lemma \ref{lem:companioninverse}, they are horizontally weakly invertible
as bordisms and thus strictly invertible as equivalence class of bordisms.
Explicitly, we find that the inverses are given by
\begin{flalign}
[M_0,\subseteq,\id]^{-1}\,=\,[M_0,\id,\subseteq]\qquad\text{and}\qquad
[N,i_1,\id]^{-1}\,=\,[N,\id,i_1]
\end{flalign}
in $\tau(\LB_m)$. Using~\eqref{eqn:composedbordism},
as well as the canonical $2$-cells in \eqref{eqn:smallercollars}, we can compute
the compositions of the right-facing bordisms in the bottom row of~\eqref{eqn:naturality for zeta: FFF --> FFF_AAA} 
and find
\begin{flalign}
\nn &[M_1,\subseteq,\id] \circ [N,\id,i_1] \circ [N,\id,\id] \circ [N,i_0,\id] \circ [M_0,\id,\subseteq] \\[4pt]
& \qquad\qquad \, =\, [N,\id,i_1]\circ [N,\id,\id]\circ [N,i_0,\id] \,=\, [N,i_0,i_1]\quad.
\end{flalign}
The commutativity of the bottom part of~\eqref{eqn:naturality for zeta: FFF --> FFF_AAA} 
then follows from the functoriality of $\FFF$.
\end{proof}

\begin{cor}\label{cor:1d}
In $m=1$ dimensions, the functor $\FFF_{(-)}^{} : \AQFT_1 \stackrel{\simeq}{\longrightarrow} 
\FFT_1^{\mathrm{t.s.}}$ defines an equivalence between the groupoid of AQFTs and the groupoid of FFTs 
satisfying the time-slice axiom.
\end{cor}
\begin{proof}
This follows directly from Theorem \ref{theo:comparison2} and the fact that $\Cau_1 = \Loc_1$, hence one has
$\AQFT_1^{\mathrm{s.g.}}=\AQFT_1$.
\end{proof}


\section{\label{sec:example}Comparing free scalar quantum field constructions}
In this section we construct the free scalar quantum field, also known
as the \textit{Klein-Gordon quantum field}, 
both as an AQFT in the sense of Definition~\ref{def:AQFT} and as an FFT in the sense of Definition~\ref{def:FFT}.
As target category we take the category  ${}^\ast\Alg_{\bbC}=
{}^{\ast}\Alg_{\mathsf{uAs}}(\Vec_{\bbC})$ of associative and unital $\ast$-algebras 
over $\bbC$. We then compare the two constructions using our Comparison Theorem~\ref{theo:comparison1}
and find that they agree. This will illustrate, through a simple example, 
our interpretation in Section~\ref{sec:comparison} that the underlying FFT of 
an AQFT captures a notion of time evolution.

\paragraph*{Preliminaries:} We start by recalling some standard 
facts about the free scalar field on globally hyperbolic Lorentzian manifolds.
Details and proofs can be found in e.g.\ \cite{BDH,BD,Bar} or in the textbook \cite{BGP}.
The classical free scalar field on $M\in\Loc_m$ is modeled by
real-valued functions $\Phi\in C^\infty(M)$ that satisfy the Klein-Gordon equation
\begin{flalign}
P_M\Phi \,:=\, \big(\Box_M + m_0^2\big)\Phi \,=\,0\quad,
\end{flalign}
where $\Box_M $ denotes the d'Alembertian and $m_0^2\geq 0$ is a fixed mass parameter.
Since $P_M=\Box_M + m_0^2$ is a normally hyperbolic differential
operator on a globally hyperbolic Lorentzian manifold $M$, it admits
unique retarded and advanced Green's operators. A retarded/advanced Green's operator
is by definition a linear map $G^\pm_M : C^\infty_{\cc}(M)\to C^\infty(M)$ from compactly supported functions
that satisfies the following properties: For all $\varphi\in C^\infty_\cc(M)$,
\begin{itemize}
\item[(1)] $P_M \,G^\pm_M(\varphi) =\varphi$,
\item[(2)] $G^\pm_M\,P_M(\varphi) =\varphi$, and
\item[(3)] $\supp \big(G^\pm_M(\varphi)\big)\subseteq J^\pm_M\big(\supp(\varphi)\big)$.
\end{itemize}
The difference $G_M := G^+_M - G^-_M : C^\infty_{\cc}(M)\to C^\infty(M)$ between
the retarded and the advanced Green's operator is called the causal propagator.
\sk

Associated with the classical free scalar field on $M$ is a Poisson 
vector space consisting of the quotient vector space
\begin{subequations}\label{eqn:LLLpoisson}
\begin{flalign}
\LLL(M)\,:=\,\frac{C^\infty_\cc(M)}{P_MC^\infty_\cc(M)}
\end{flalign}
and the linear Poisson structure
\begin{flalign}
\nn \tau_M \,:\, \LLL(M)\otimes \LLL(M)~&\longrightarrow~\bbR\quad,\\
[\varphi_1]\otimes[\varphi_2]~&\longmapsto~\int_M \varphi_1\,G_M(\varphi_2)\,\vol_M\quad,
\end{flalign}
\end{subequations}
where $\vol_M$ denotes the canonical volume form that is determined 
by the metric and orientation of $M\in\Loc_m$. The interpretation of
$\LLL(M)$ is that of linear functions (i.e.\ observables) 
on the solution space $\Sol(M) := \{\Phi\in C^\infty(M)\,:\,P_M\Phi =0\}$,
which are evaluated by using the non-degenerate integration 
pairing $\LLL(M)\otimes \Sol(M)\to \bbR\,,~[\varphi]\otimes \Phi\mapsto \int_M \varphi\,\Phi\,\vol_M$.
\sk

There exists an alternative description of the Poisson vector space \eqref{eqn:LLLpoisson}
in terms of the vector space
\begin{subequations}\label{eqn:Solpoisson}
\begin{flalign}
\Sol_{\sc}(M)\,:=\, \Big\{\Phi\in C^\infty_{\sc}(M)\,:\, P_M\Phi =0\Big\}
\end{flalign}
of solutions with space-like compact support and the linear Poisson structure
\begin{flalign}
\nn \sigma_M\,:\, \Sol_{\sc}(M)\otimes \Sol_{\sc}(M)~&\longrightarrow~\bbR\quad,\\
\Phi_1\otimes \Phi_2 ~&\longmapsto~\int_{\Sigma} \big(\Phi_1\,\nabla_{n}\Phi_2 - \Phi_2\,\nabla_n \Phi_1\big)\,\vol_{\Sigma}\quad,
\end{flalign}
\end{subequations}
where $\Sigma\subset M$ is an arbitrary choice of Cauchy surface, 
$\vol_\Sigma$ denotes the canonical volume form on $\Sigma$ 
that is determined by the metric, orientation and time-orientation of $M\in\Loc_m$, 
and $n\in \Gamma^\infty(TM\vert_{\Sigma})$ is the future-pointing unit normal vector field of $\Sigma$.
It is important to stress that the Poisson structure $\sigma_M$ in \eqref{eqn:Solpoisson}
does not depend on the choice of Cauchy surface $\Sigma\subset M$.
The causal propagator defines a canonical isomorphism
\begin{flalign}
G_M\,:\, \LLL(M)~\stackrel{\cong}{\longrightarrow}~\Sol_{\sc}(M)
\end{flalign}
of Poisson vector spaces, i.e.\ $\sigma_M\big(G_M[\varphi_1],G_M[\varphi_2]\big) = \tau_M\big([\varphi_1],[\varphi_2]\big)$,
for all $[\varphi_1],[\varphi_2]\in\LLL(M)$.
\sk

In addition to the existence of retarded/advanced Green's operators,
the normally hyperbolic differential operator $P_M$ also has a well-posed
initial value problem on every globally hyperbolic Lorentzian manifold $M$.
The compactly supported initial data on a Cauchy surface $\Sigma\subset M$
for such second-order partial differential equation
are described by the vector space
\begin{subequations}\label{eqn:Datapoisson}
\begin{flalign}
\Data_\cc(\Sigma)\,:=\, C^\infty_\cc(\Sigma)^{\oplus 2}\quad,
\end{flalign}
which carries a canonical linear Poisson structure
\begin{flalign}
\nn \tau_{\Sigma}\,:\, \Data_\cc(\Sigma)\otimes  \Data_\cc(\Sigma)~&\longrightarrow~\bbR\quad,\\
(\phi_1,\pi_1)\otimes (\phi_2,\pi_2)~&\longmapsto~\int_{\Sigma}\big(\phi_1\,\pi_2 - \phi_2\,\pi_1\big)\,\vol_\Sigma\quad.
\end{flalign}
\end{subequations}
Well-posedness of the initial value problem then implies that the linear map
\begin{flalign}\label{eqn:solvemap}
\res_{(M,\Sigma)}\,:\, \Sol_{\sc}(M)~\stackrel{\cong}{\longrightarrow}~\Data_\cc(\Sigma)~~,\quad
\Phi~\longmapsto~\big(\Phi\vert_{\Sigma},\nabla_n\Phi\vert_{\Sigma}\big)
\end{flalign}
that assigns to space-like compact solutions their initial data is an isomorphism. 
Due to the independence of $\sigma_M$ on the choice of Cauchy surface, 
this isomorphism preserves the Poisson structures, i.e.\ 
$\tau_{\Sigma}\big(\res_{(M,\Sigma)}\Phi_1,\res_{(M,\Sigma)}\Phi_2\big) 
= \sigma_M\big(\Phi_1,\Phi_2\big)$, for all $\Phi_1,\Phi_2\in \Sol_{\sc}(M)$.
\sk

Summing up, we have presented three isomorphic descriptions
\begin{flalign}\label{eqn:Poissoniso}
\xymatrix@C=3.5em{
\LLL(M) \ar[r]^-{G_M}_-{\cong}&\Sol_{\sc}(M)
\ar[r]^-{\res_{(M,\Sigma)}}_-{\cong}& \Data_\cc(\Sigma)
}
\end{flalign}
of the Poisson vector space associated with the free scalar field on a globally hyperbolic
Lorentzian manifold $M\in\Loc_m$ with a choice of Cauchy surface $\Sigma\subset M$.
Note that the first two descriptions do not depend
on the Cauchy surface $\Sigma\subset M$, while the 
third description does not depend on the ambient space $M$ of $\Sigma$.
\sk

Finally, we recall that the canonical commutation relation (CCR) 
quantization of a Poisson vector space $V=(V,\tau)$ 
is the associative and unital $\ast$-algebra 
\begin{flalign}
\CCR(V)\,:=\, T^{\otimes}_\bbC V \big/I_{\tau}\,\in\,\Alg
\end{flalign}
that is defined as the quotient of the free
$\ast$-algebra $T^\otimes_\bbC V := 
\bigoplus_{n\geq 0} \big(V^{\otimes n}\otimes\bbC\big)$ by the two-sided $\ast$-ideal $I_{\tau}$
generated by the commutation relations $v_1\otimes v_2 - v_2\otimes v_1 - \ii\,\tau(v_1,v_2)\,\oone$,
for all $v_1,v_2\in V$. This quantization construction is functorial 
$\CCR : \mathbf{PoVec} \to \Alg$ on the category of Poisson vector spaces
and Poisson structure preserving linear maps.

\paragraph*{The AQFT construction:} The standard construction
of the free scalar quantum field as an AQFT uses the first
of the three isomorphic descriptions from \eqref{eqn:Poissoniso}.
One interprets $\LLL(M)$ as the linear classical observables
of the theory and constructs an algebra of quantum observables on $M\in \Loc_m$
in terms of CCR quantization
\begin{flalign}\label{eqn:KGAQFTalgebra}
\AAA_{\KG}(M)\,:=\,\CCR\big(\LLL(M)\big)\,\in\,\Alg\quad.
\end{flalign}
Given any morphism $f : M\to M^\prime$ in $\Loc_m$, push-forward of compactly
supported functions defines a morphism $\LLL(f) : \LLL(M)\to \LLL(M^\prime)$
of Poisson vector spaces. Applying the $\CCR$-functor then yields an $\Alg$-morphism
\begin{flalign}\label{eqn:KGAQFTmorphisms}
\AAA_{\KG}(f)\,:\, \AAA_{\KG}(M)~\longrightarrow~\AAA_{\KG}(M^\prime)\quad.
\end{flalign}
One easily checks that the assignment $\AAA_{\KG} : \Loc_m\to \Alg$
given by \eqref{eqn:KGAQFTalgebra} and \eqref{eqn:KGAQFTmorphisms} defines
a functor that satisfies the axioms of an AQFT from Definition~\ref{def:AQFT}.
This completes the construction of the free scalar quantum field
\begin{flalign}\label{eqn:KGAQFT}
\AAA_{\KG}\,\in\,\AQFT_m
\end{flalign}
as an AQFT.

\paragraph*{The FFT construction:} To construct the free scalar quantum field
as an FFT in the sense of Definition~\ref{def:FFT}, we use Corollary~\ref{cor:FFTsimple}
and the third of the three isomorphic descriptions from \eqref{eqn:Poissoniso}.
To an object $(M,\Sigma)\in\tau(\LB_m)$, we assign the algebra
\begin{flalign}\label{eqn:KGFFTalgebra}
\FFF_{\KG}(M,\Sigma)\,:=\, \CCR\big(\Data_\cc(\Sigma)\big)\,\in\,\Alg
\end{flalign}
that is obtained by CCR quantization of the Poisson vector space of compactly supported
initial data on the marked Cauchy surface $\Sigma\subset M$.
Consider now any morphism $\big[(N,i_0,i_1) : (M_0,\Sigma_0)\hto (M_1,\Sigma_1)\big]$ in $\tau(\LB_m)$,
i.e.\ an equivalence class under globular $2$-cells of a globally hyperbolic Lorentzian
bordism \eqref{eqn:LBord1cell} represented by the Cauchy morphisms
\begin{flalign}
\xymatrix{
M_0 & \ar[l]_-{\subseteq} V_0 \ar[r]^-{i_0}& N & \ar[l]_-{i_1}V_1 \ar[r]^-{\subseteq}& M_1
}\quad.
\end{flalign}
Observing that the Poisson vector spaces  in \eqref{eqn:Datapoisson} 
only depend on the Cauchy surface and not on its collar region, we obtain two identities
\begin{subequations}
\begin{flalign}
\FFF_{\KG}(M_0,\Sigma_0)\,=\, \FFF_{\KG}(V_0,\Sigma_0)\quad,\qquad
\FFF_{\KG}(M_1,\Sigma_1)\,=\, \FFF_{\KG}(V_1,\Sigma_1)
\end{flalign}
and two isomorphisms
\begin{flalign}
\CCR(\Data_\cc(i_0\vert))\,:\,\FFF_{\KG}(V_0,\Sigma_0)~&\stackrel{\cong}{\longrightarrow}~\FFF_{\KG}(N,i_0(\Sigma_0))\quad,\\
\CCR(\Data_\cc(i_1\vert))\,:\,\FFF_{\KG}(V_1,\Sigma_1)~&\stackrel{\cong}{\longrightarrow}~\FFF_{\KG}(N,i_1(\Sigma_1))\quad.
\end{flalign}
\end{subequations}
The latter are induced by push-forward of compactly supported functions
along the orientation-preserving diffeomorphisms 
$i_{0/1}\vert : \Sigma_{0/1}\stackrel{\cong}{\longrightarrow} \iota_{0/1}(\Sigma_{0/1})$
that are obtained by restricting and
co-restricting $i_{0/1} : V_{0/1}\to N$ to the Cauchy surfaces. 
We then define the $\Alg$-morphism
\begin{subequations}\label{eqn:KGFFTbordism}
\begin{flalign}
\FFF_{\KG}\big([N,i_0,i_1]\big)\,:\,\FFF_{\KG}(M_0,\Sigma_0)~\longrightarrow~\FFF_{\KG}(M_1,\Sigma_1)
\end{flalign}
by the commutative diagram
\begin{flalign}
\begin{gathered}
\resizebox{0.9\hsize}{!}{
$\xymatrix@C=7em{
\ar@{=}[d]\FFF_{\KG}(M_0,\Sigma_0) \ar@{-->}[rr]^-{\FFF_{\KG}([N,i_0,i_1])}&& \FFF_{\KG}(M_0,\Sigma_0) \ar@{=}[d]\\
\ar[d]_-{\CCR(\Data_\cc(i_0\vert))}^-{\cong}\FFF_{\KG}(V_0,\Sigma_0)&&\FFF_{\KG}(V_1,\Sigma_1) \ar[d]^-{\CCR(\Data_\cc(i_1\vert))}_-{\cong}\\
\FFF_{\KG}(N,i_0(\Sigma_0))&\ar[l]_-{\cong}^-{\CCR(\res_{(N,\iota_0(\Sigma_0))})}\CCR(\Sol_\sc(N))
\ar[r]^-{\cong}_-{\CCR(\res_{(N,\iota_1(\Sigma_1))})}&\FFF_{\KG}(N,i_1(\Sigma_1))
}
$}
\end{gathered}
\end{flalign}
\end{subequations}
of $\Alg$-isomorphisms. Note that the bottom horizontal zig-zag
uses the Poisson vector space isomorphisms \eqref{eqn:solvemap}
that are obtained from well-posedness of the initial value problem.
The interpretation of the $\Alg$-isomorphism $\FFF_{\KG}\big([N,i_0,i_1]\big)$ is as follows:
As input, it takes observables on the Cauchy surface $\Sigma_0\subset M_0$
and identifies them via $i_0\vert$ with observables on the 
diffeomorphic Cauchy surface $i_0(\Sigma_0)\subset N$. These are then evolved
in time through the bordism $N$
via the well-posed initial value problem to a Cauchy surface $i_1(\Sigma_1)\subset N$ in the causal future,
and then identified via $i_1\vert$ with observables on the Cauchy surface $\Sigma_1\subset M_1$.
Hence, the $\Alg$-morphism $\FFF_{\KG}\big([N,i_0,i_1]\big)$ encodes a notion of time evolution
along the bordism.
\sk

One easily checks that \eqref{eqn:KGFFTbordism} does not depend on
the choice of representative of the equivalence class 
$\big[(N,i_0,i_1) : (M_0,\Sigma_0)\hto (M_1,\Sigma_1)\big]$ defining a morphism in $\tau(\LB_m)$
and that the assignment $\FFF_{\KG} : \tau(\LB_m)\to \Alg$ given by \eqref{eqn:KGFFTalgebra} 
and \eqref{eqn:KGFFTbordism} is a functor. This defines the free scalar quantum field
\begin{flalign}\label{eqn:KGFFT}
\FFF_{\KG}\,\in\,\FFT_m^{\mathrm{t.s.}}
\end{flalign}
as an FFT.

\paragraph*{Comparison:} Applying Theorem~\ref{theo:comparison1} to the AQFT in \eqref{eqn:KGAQFT} that describes
the free scalar quantum field, we obtain $\FFF_{\AAA_{\KG}}\in\FFT_m^{\mathrm{t.s.}}$ that assigns to
an object $(M,\Sigma)\in \tau(\LB_m)$ the algebra
\begin{flalign}
\FFF_{\AAA_{\KG}}(M,\Sigma) \,=\, \AAA_{\KG}(M)\,=\, \CCR\big(\LLL(M)\big)\,\in\,\Alg\quad.
\end{flalign}
Using the composite of the Poisson vector space isomorphism \eqref{eqn:Poissoniso},
we obtain an $\Alg$-isomorphism
\begin{flalign}\label{eqn:scalarcomparison}
\CCR\big(\res_{(M,\Sigma)}\,G_M\big)\,:\, \FFF_{\AAA_{\KG}}(M,\Sigma)
~\stackrel{\cong}{\longrightarrow}~\FFF_{\KG}(M,\Sigma)\quad,
\end{flalign}
for all objects $(M,\Sigma)\in \tau(\LB_m)$, to the values of the alternative 
FFT \eqref{eqn:KGFFT} for the free scalar quantum field.
\begin{propo}\label{prop:scalarcomparison}
The components \eqref{eqn:scalarcomparison} define an $\FFT_m$-isomorphism
$\FFF_{\AAA_{\KG}}\stackrel{\cong}{\longrightarrow} \FFF_{\KG}$ between 
the FFT obtained by applying Theorem~\ref{theo:comparison1} to the AQFT 
in \eqref{eqn:KGAQFT} and the FFT in \eqref{eqn:KGFFT}. Hence, 
the AQFT and FFT constructions of the free scalar quantum field 
are compatible with each other.
\end{propo}
\begin{proof}
We have to show that the components \eqref{eqn:scalarcomparison} are natural with respect to all
morphisms $\big[(N,i_0,i_1) : (M_0,\Sigma_0)\hto (M_1,\Sigma_1)\big]$ in $\tau(\LB_m)$.
Recalling \eqref{eqn:FFF_AAAmorphism} and \eqref{eqn:KGFFTbordism}, this 
follows directly by applying the CCR-functor  $\CCR : \mathbf{PoVec}\to \Alg$
to the commutative diagram
\begin{flalign}
\begin{gathered}
\resizebox{0.9\hsize}{!}{
$\xymatrix@R=3em@C=4em{
\Data_\cc(\Sigma_0) \ar[r]^-{\Data_\cc(i_0\vert)}& \Data_\cc(i_0(\Sigma_0)) & &\Data_\cc(i_1(\Sigma_1)) & \ar[l]_-{\Data_\cc(i_1\vert)}\Data_\cc(\Sigma_1)\\
\ar[u]^-{~~\res_{(M_0,\Sigma_0)}} \Sol_{\sc}(M_0)&\ar[l]^-{\cong} \ar[lu]_-{~~\res_{(V_0,\Sigma_0)}}\Sol_{\sc}(V_0)
\ar[r]_-{\Sol_{\sc}(i_0)}&\ar[lu]_-{\res_{(N,i_0(\Sigma_0))}}\Sol_\sc(N)
\ar[ru]^-{\res_{(N,i_1(\Sigma_1))}~~}&\ar[l]^-{\Sol_{\sc}(i_1)}\Sol_{\sc}(V_1)\ar[r]_-{\cong} \ar[ru]^-{\res_{(V_1,\Sigma_1)}~~}&\ar[u]_-{\res_{(M_1,\Sigma_1)}} \Sol_{\sc}(M_1)\\
\ar[u]^-{G_{M_0}} \LLL(M_0) &\ar[l]^-{\cong} \ar[u]^-{G_{V_0}}\LLL(V_0) \ar[r]_-{\LLL(i_0)} & \ar[u]^-{G_N}\LLL(N) &\ar[l]^-{\LLL(i_1)} \LLL(V_1)\ar[u]^-{G_{V_1}}\ar[r]_-{\cong}& \LLL(M_1)\ar[u]_-{G_{M_1}}
}
$
}
\end{gathered}
\end{flalign}
of isomorphisms in the category of Poisson vector spaces.
The functorial structure on $\Sol_{\sc} : \Loc_m\to \mathbf{PoVec}$ 
is given by extending space-like compact solutions along $\Loc_m$-morphisms
via the well-posed initial value problem.
\end{proof}


\section*{Acknowledgments}
We would like to thank Christoph Schweigert for useful comments and for
pointing us to his ongoing work on double categories in conformal quantum field theory.
The research of S.B.\ is funded by the Deutsche Forschungsgemeinschaft 
(DFG, German Research Foundation) under the project number 468806966.
J.M.\ is funded by an EPSRC PhD scholarship (2742043) of the School of Mathematical Sciences at
the University of Nottingham.
A.S.\ gratefully acknowledges the support of 
the Royal Society (UK) through a Royal Society University 
Research Fellowship (URF\textbackslash R\textbackslash 211015)
and Enhancement Grants (RF\textbackslash ERE\textbackslash 210053 and 
RF\textbackslash ERE\textbackslash 231077).


\appendix

\section{\label{app:pseudocats}A $2$-adjunction between pseudo-categories and categories}
The goal of this appendix is to upgrade Constructions
\ref{constr:iota} and \ref{constr:tau} to $2$-functors
and to prove Theorem~\ref{theo:2adjunction}.

\paragraph{The $2$-functor $\iota$:} We define the $2$-functor $\iota: \Cat^{(2,1)}\to \PsCat^{\mathrm{fib}}$ 
on objects as explained in Construction \ref{constr:iota}. Given any functor
$F : \CC\to \DD$ between two ordinary categories, we define the pseudo-functor
$\iota(F) : \iota(\CC)\to \iota(\DD)$ by the following data as in Definition~\ref{def:PsFunctor}:
\begin{itemize}
\item[(i)] The functor $\iota(F)_0 := F : \iota(\CC)_0=\mathrm{core}(\CC)\to\mathrm{core}(\DD)=\iota(\DD)_0$
is given by restricting the original functor $F$ to the cores. The functor
$\iota(F)_1 : \iota(\CC)_1\to \iota(\DD)_1$ is defined by applying the original functor $F$ to commutative squares
\begin{flalign}
\begin{gathered}
\xymatrix{
c_0^\prime \ar[r]^-{f^\prime} & c_1^\prime\\
\ar[u]_-{\cong}^-{g_0}c_0\ar[r]_-{f}&c_1\ar[u]^-{\cong}_-{g_1}
}
\end{gathered}
~~\stackrel{\iota(F)_1}{\longmapsto}~~
\begin{gathered}
\xymatrix{
F(c_0^\prime) \ar[r]^-{F(f^\prime)} & F(c_1^\prime)\\
\ar[u]_-{\cong}^-{F(g_0)}F(c_0)\ar[r]_-{F(f)}&F(c_1)\ar[u]^-{\cong}_-{F(g_1)}
}
\end{gathered}\qquad.
\end{flalign}

\item[(ii)] The natural isomorphisms $\iota(F)^{\cmp}$ and $\iota(F)^u$ are trivial,
i.e.\ the identity natural transformations.
\end{itemize}
Given any natural isomorphism $\zeta : F\Rightarrow G$ between two ordinary
functors $F,G: \CC\to\DD$, we define the transformation $\iota(\zeta) : \iota(F)\Rightarrow \iota(G)$
between the two pseudo-functors $\iota(F),\iota(G) : \iota(\CC)\to\iota(\DD)$ by the following 
data as in Definition~\ref{def:Transformation}: The natural transformation
$\iota(\zeta)^0 : \iota(F)_0 \Rightarrow \iota(G)_0$ is defined by the components
\begin{flalign}
\iota(\zeta)^0 \,:=\, \left\{
\begin{gathered}
\xymatrix@R=0.25em@C=0.25em{
G(c)\\
~\\
\ar[uu]_-{\cong}^-{\zeta_c} F(c)
}
\end{gathered}~~:~~c\in \CC
\right\}\quad,
\end{flalign}
which is well-defined because $\zeta$ is by hypothesis a natural \textit{isomorphism}.
The natural transformation $\iota(\zeta)^1 : \iota(F)_1 \Rightarrow \iota(G)_1$ is defined by the components
\begin{flalign}
\iota(\zeta)^1 \,:=\, \left\{
\begin{gathered}
\xymatrix@R=0.25em@C=0.25em{
G(c_0)\ar[rr]^-{G(f)}&&G(c_1)\\
&~&\\
\ar[uu]_-{\cong}^-{\zeta_{c_0}} F(c_0) \ar[rr]_-{F(f)}&& F(c_1)\ar[uu]^-{\cong}_-{\zeta_{c_1}}
}
\end{gathered}~~:~~\big(f: c_0\to c_1\big)\in \CC
\right\}\quad.
\end{flalign}
The conditions (1), (2) and (3) from Definition~\ref{def:Transformation} clearly hold true.

\paragraph{The $2$-functor $\tau$:} We define the $2$-functor $\tau: \PsCat^{\mathrm{fib}}\to  \Cat^{(2,1)}$ 
on objects as explained in Construction \ref{constr:tau}. Given any pseudo-functor $F : \C\to\D$, we define
the ordinary functor
\begin{flalign}
\nn \tau(F) \,:\,\tau(\C)~&\longrightarrow~\tau(\D)\quad,\\
\nn c ~&\longmapsto~F_0(c)\quad,\\
[f: c_0\hto c_1]~&\longmapsto~\big[F_1(f) : F_0(c_0)\hto F_0(c_1)\big]\quad. 
\end{flalign}
This is well-defined on equivalence classes because $F_1$ maps globular $2$-cells in
$\C$ to globular $2$-cells in $\D$. The functor $\tau(F)$ preserves compositions and identities
because the coherence natural isomorphisms $F^{\cmp}$ and $F^u$ have globular components,
hence they are trivial at the level of equivalence classes.
\sk

Given any transformation $\zeta : F\Rightarrow G$ between two pseudo-functors
$F,G : \C\to \D$, we define the natural isomorphism $\tau(\zeta) : \tau(F)\Rightarrow\tau(G)$ 
between the two ordinary functors $\tau(F),\tau(G) : \tau(\C)\to \tau(\D)$ by the components
\begin{flalign}
\tau(\zeta)\,:=\, \left\{\big[\hat{\zeta}^0_c : F_0(c)\hto G_0(c) \big]~~:~~ c\in \C_0\right\}
\end{flalign}
given by an arbitrary choice of companions for the components $\zeta^0_c : F_0(c)\to G_0(c)$
of $\zeta^0$. By \cite[Lemma 3.8]{Shulman}, different choices of companions define the 
same equivalence class, hence the components of $\tau(\zeta)$ are well-defined.
To prove that these components are natural, i.e.\ 
$[G_1(f)]\,[\hat{\zeta}^0_{c_0}] = [\hat{\zeta}^0_{c_1}] \, [F_1(f)]$ for all
horizontal morphisms $f : c_0\hto c_1$ in $\C$, we compose the $2$-cell
component $\zeta^1_f$ of the natural transformation $\zeta^1: F_1\Rightarrow G_1$ 
with the $2$-cells from the Definition~\ref{def:fibrant} for the companions
according to
\begin{flalign}
\begin{gathered}
\xymatrix@R=0.25em@C=1em{
F_0(c_0) \ar[rr]|{\sla}^-{\hat{\zeta}^0_{c_0}} && G_0(c_0) \ar[rr]|{\sla}^-{G_1(f)}&& G_0(c_1)  \ar[rr]|{\sla}^-{u(G_0(c_1))}&& G_0(c_1)\\
&\rotatebox[origin=c]{90}{$\Rightarrow$}& &{\scriptstyle \zeta^1_f}~\rotatebox[origin=c]{90}{$\Rightarrow$}& &\rotatebox[origin=c]{90}{$\Rightarrow$}& \\
\ar@{=}[uu] F_0(c_0) \ar[rr]|{\sla}_-{u(F_0(c_0))} && F_0(c_0) \ar[uu]^-{\zeta^0_{c_0}}\ar[rr]|{\sla}_-{F_1(f)}&&  \ar[uu]_-{\zeta^0_{c_1}} F_0(c_1)\ar[rr]|{\sla}_-{\hat{\zeta}^0_{c_1}} && G_0(c_1)\ar@{=}[uu] 
}
\end{gathered}\qquad.
\end{flalign}
Using the unitors $\mathsf{l}_\D$ and $\mathsf{r}_\D$, we obtain a globular $2$-cell
which, when passing to equivalence classes, exhibits the naturality of $\tau(\zeta)$.

\paragraph{The $2$-adjunction counit:} We define a $2$-natural transformation
$\varepsilon : \tau\,\iota \Rightarrow \id_{\Cat^{(2,1)}}$ that will serve as the counit
for our $2$-adjunction \eqref{eqn:tauiotaadjunction}. 
Given any ordinary category $\CC\in\Cat^{(2,1)}$, one easily checks
that $\tau\iota(\CC) = \CC$ since there are no non-trivial globular $2$-cells in $\iota(\CC)$.
Similarly, one checks that $\tau\iota(F) =F$ for all functors and $\tau\iota(\zeta) = \zeta$
for all natural isomorphisms. This implies that
\begin{flalign}
\tau\,\iota \,=\, \id_{\Cat^{(2,1)}}
\end{flalign}
as $2$-functors, and hence 
we shall take for $\varepsilon = \mathrm{Id}$ the identity $2$-natural transformation.

\paragraph{The $2$-adjunction unit:} We define a $2$-natural transformation
$\eta : \id_{\PsCat^{\mathrm{fib}}} \Rightarrow \iota\,\tau$ that will serve as the unit
for our $2$-adjunction \eqref{eqn:tauiotaadjunction}.
For any fibrant pseudo-category $\C\in\PsCat^{\mathrm{fib}}$, we define
the pseudo-functor $\eta_\C : \C\to \iota\tau(\C)$ by specifying the data from
Definition~\ref{def:PsFunctor} as follows:
\begin{itemize}
\item[(i)] The assignments
\begin{flalign}
\nn(\eta_\C)_0 \,: \C_0~&\longrightarrow~\iota\tau(\C)_0 = \mathrm{core}(\tau(\C))\quad,\\
\nn c~&\longmapsto~ c\quad,\\
(g:c\to c^\prime)~&\longmapsto~[\hat{g} : c\hto c^\prime] \,=:\, \big([\hat{g}] : c\to c^\prime\big)
\end{flalign}
and
\begin{flalign}
\nn(\eta_\C)_1 \,: \C_1~&\longrightarrow~\iota\tau(\C)_1\quad,\\
\nn (f:c_0\hto c_1)~&\longmapsto~ [f:c_0\hto c_1]\,=:\, \big([f] : c_0\to c_1\big)\quad,\\
\begin{gathered}
\xymatrix@R=0.25em@C=0.25em{
c_0^\prime \ar[rr]|{\sla}^-{f^\prime}&& c_1^\prime\\
&{\scriptstyle \alpha}~\rotatebox[origin=c]{90}{$\Rightarrow$}&\\
\ar[uu]^-{g_0} c_0  \ar[rr]|{\sla}_-{f}  && c_1\ar[uu]_-{g_1} 
}
\end{gathered}
~~&\longmapsto~~
\begin{gathered}
\xymatrix@R=0.75em@C=0.75em{
c_0^\prime \ar[rr]^-{[f^\prime]}&& c_1^\prime\\
&& \\
\ar[uu]^-{[\hat{g}_0]} c_0  \ar[rr]_-{[f]}  && c_1\ar[uu]_-{[\hat{g}_1]} 
}
\end{gathered}
\end{flalign}
are functors as a consequence of \cite[Lemmas 3.12 and 3.13]{Shulman}.
Commutativity of the square in $\tau(\C)$ follows by composing the $2$-cell
$\alpha$ with the $2$-cells from the Definition~\ref{def:fibrant} for the companions
according to
\begin{flalign}
\begin{gathered}
\xymatrix@R=0.25em@C=0.25em{
c_0 \ar[rr]|{\sla}^-{\hat{g}_0}&& c_0^\prime \ar[rr]|{\sla}^-{f^\prime}&& c_1^\prime \ar[rr]|{\sla}^-{u(c_1^\prime)}&& c_1^\prime \\
&~\rotatebox[origin=c]{90}{$\Rightarrow$}~&  &{\scriptstyle \alpha}~\rotatebox[origin=c]{90}{$\Rightarrow$}&  &~\rotatebox[origin=c]{90}{$\Rightarrow$}~&\\
\ar@{=}[uu] c_0 \ar[rr]|{\sla}_-{u(c_0)}&& \ar[uu]^-{g_0} c_0  \ar[rr]|{\sla}_-{f}  && c_1\ar[uu]_-{g_1} 
\ar[rr]|{\sla}_-{\hat{g}_1}&& c_1^\prime\ar@{=}[uu]
}
\end{gathered}
\end{flalign}
and then passing to equivalence classes.

\item[(ii)] The natural isomorphisms $(\eta_\C)^{\cmp}$ and $(\eta_\C)^u$ are trivial.
\end{itemize}
Using also \cite[Lemma 3.16]{Shulman}, one checks that the components
$\eta_\C : \C\to \iota\tau(\C)$ define a $2$-natural transformation
$\eta : \id_{\PsCat^{\mathrm{fib}}} \Rightarrow \iota\,\tau$.

\paragraph{The triangle identities:} Since $\varepsilon = \mathrm{Id}$, the triangle identities
reduce to showing that $\eta_{\iota(\CC)}=\id$, for all ordinary categories $\CC\in\Cat^{(2,1)}$,
and that $\tau(\eta_{\C}) = \id$, for all fibrant pseudo-categories $\C\in\PsCat^{\mathrm{fib}}$.
These are elementary checks that can be carried out using the explicit formulas from this appendix.
This completes the proof of Theorem~\ref{theo:2adjunction}.


%
%


\end{document}